%% file: AAMAS-2026-Formatting-Instructions-CCBY/AAMAS_Camera_ready.tex
\documentclass[sigconf]{aamas}

\usepackage{balance} 

\usepackage[utf8]{inputenc} 
\usepackage[T1]{fontenc}    
\usepackage{hyperref}       
\usepackage{url}            
\usepackage{booktabs}       
\usepackage{amsfonts}       
\usepackage{nicefrac}       
\usepackage{microtype}      
\usepackage{xcolor}         

\usepackage{enumitem}

\usepackage{graphicx}
\usepackage{subfigure}
\usepackage{booktabs} 

\usepackage{wrapfig}
\usepackage{algorithm}
\usepackage{algorithmic}

\usepackage{amsmath}
\usepackage{mathtools}
\usepackage{amsthm}

\usepackage[capitalize,noabbrev]{cleveref}

\usepackage[most]{tcolorbox}

\newtcolorbox{shadowboxtext}{
  colback=gray!10,       
  boxrule=1pt,           
  frame hidden,          
  arc=0pt,               
  drop shadow,           
  left=5pt,              
  right=5pt,
  top=5pt,
  bottom=5pt,
}

\theoremstyle{plain}
\newtheorem{theorem}{Theorem}[section]

\newtheorem{lemma}[theorem]{Lemma}

\theoremstyle{definition}
\newtheorem{definition}[theorem]{Definition}
\newtheorem{assumption}[theorem]{Assumption}
\theoremstyle{remark}
\newtheorem{remark}[theorem]{Remark}

\setcopyright{ifaamas}
\acmConference[AAMAS '26]{Proc.\@ of the 25th International Conference
on Autonomous Agents and Multiagent Systems (AAMAS 2026)}{May 25 -- 29, 2026}
{Paphos, Cyprus}{C.~Amato, L.~Dennis, V.~Mascardi, J.~Thangarajah (eds.)}
\copyrightyear{2026}
\acmYear{2026}
\acmDOI{}
\acmPrice{}
\acmISBN{}

\acmSubmissionID{678}

\title[AAMAS-2026 Formatting Instructions]{Global Convergence to Nash Equilibrium in Nonconvex General-Sum Games under the $n$-Sided PL Condition}

\author{Yutong Chao}
\affiliation{
  \institution{Technical University of Munich}
  \city{Munich}
  \country{Germany}}
\email{yutong.chao@tum.de}

\author{Jalal Etesami}
\affiliation{
  \institution{Technical University of Munich}
  \institution{Munich Institute of Robotics and Machine Intelligence}
  \city{Munich}
  \country{Germany}}
\email{j.etesami@tum.de}

\begin{abstract}
We consider the problem of finding a Nash equilibrium (NE) in a general-sum game, where player $i$'s objective is $f_i(x)=f_i(x_1,...,x_n)$, with $x_j\in\mathbb{R}^{d_j}$ denoting the strategy variables of player $j$. Our focus is on investigating first-order gradient-based algorithms and their variations, such as the block coordinate descent (BCD) algorithm, for tackling this problem. We introduce a set of conditions, called the $n$-sided PL condition, which extends the well-established gradient dominance condition a.k.a Polyak-{\L}ojasiewicz (PL) condition and the concept of multi-convexity. This condition, satisfied by various classes of non-convex functions, allows us to analyze the convergence of various gradient descent (GD) algorithms. 
Moreover, our study delves into scenarios where the standard gradient descent methods fail to converge to NE.
In such cases, we propose adapted variants of GD that converge towards NE and analyze their convergence rates. Finally, we evaluate the performance of the proposed algorithms through several experiments.
\end{abstract}


\newcommand{\BibTeX}{\rm B\kern-.05em{\sc i\kern-.025em b}\kern-.08em\TeX}

\begin{document}


\pagestyle{fancy}
\fancyhead{}


\maketitle 


\input{AAMAS-2026-Formatting-Instructions-CCBY/1_introduction}
\input{AAMAS-2026-Formatting-Instructions-CCBY/2_related}

\input{AAMAS-2026-Formatting-Instructions-CCBY/3_definitions}

\input{AAMAS-2026-Formatting-Instructions-CCBY/6_algorithms}

\input{AAMAS-2026-Formatting-Instructions-CCBY/7_applications}

\input{AAMAS-2026-Formatting-Instructions-CCBY/8_conclusion}



\bibliographystyle{ACM-Reference-Format} 
\bibliography{sample}

\onecolumn
\input{appendix}

\end{document}

%% file: AAMAS-2026-Formatting-Instructions-CCBY/1_introduction.tex
\section{Introduction}

Optimization problems with nonconvex objectives appear in many applications from game theory to machine learning \citep{scutari2010convex,zhao2022provably,cai2024near}, such as training deep neural networks \citep{heaton2018ian,reyad2023modified,elkabetz2021continuous}, or policy optimization in reinforcement learning \citep{silver2014deterministic,zhao2022provably,yang2022constrained}. {
In practice, the majority of these applications involve interacting agents whose optimal behavior emerges from solving nonconvex optimization problems.
Despite the notable empirical advancements, there is a lack of understanding of the theoretical convergence guarantees of existing multi-agent methods.
Even the fully cooperative settings remain largely underexplored and constitute one of the current frontiers in machine learning research. }

From a game-theoretic perspective, we focus on a multi-agent general-sum game where the $i$-th player's objective function $f_i$ captures the aggregate impact of agents' strategies, which is parameterized by $ \{x_i\}_{i=1} ^n$. Agents aim to optimize their individual objective function.
The most common solution concept for the underlying optimization problem is Nash equilibrium (NE) \citep{nash1950bargaining}. 
Along this line of research, agents are said to be in a Nash equilibrium if no agent can gain by unilaterally deviating from its current strategy, assuming that all other agents maintain theirs. 
We are interested in finding Nash equilibrium $x^\star=(x_1^\star,\cdots,x_n^\star)\in\mathbb R^d$ of the nonconvex functions $\{f_i(x): i\in[n]\}$, i.e.,
\begin{equation}\label{min_problem}
        f_i(x_i^\star;x_{-i}^\star)\leq f_i(y_i;x_{-i}^\star),\quad \forall i\in[n], \forall y_i\in\mathbb{R}^{d_i},
\end{equation}
where $x_{-i}^\star$ denotes all blocks except $i$-th block and $d=\sum_id_i$. 
Note that in the special case where $f_1(x)=\cdots=f_n(x)$, the setting reduces to a potential game in which all players jointly aim to minimize a common objective function. 
It is known that finding the NE is PPAD-complete in general \citep{daskalakis2009complexity}. 
Thus, there is a need for additional structural assumptions to reduce the computational complexity of finding NE.

For example, in extensive-form games, which extend normal-form interactions to multi-stage decision processes, equilibria can often be analyzed using backward induction or dynamic programming–based techniques \cite{celli2019learning,lee2021last}. In linear-quadratic (LQ) games, where player payoffs are quadratic functions governed by linear dynamical systems, the algebraic structure permits explicit characterizations of NE via coupled Riccati equations \cite{hambly2023policy,fazel2018global,zhang2019policy,roudneshin2020reinforcement}. However, these guarantees are inherently tied to the particular game models under consideration, and therefore do not generalize to the broader class of general-sum games with fewer structural assumptions and nonconvex payoff functions, where the existence and computability of equilibrium remain much less understood.


To tackle a nonconvex optimization problem, a straightforward approach is to introduce additional structural assumptions to achieve convergence guarantees. Within this scope, various relaxations of convexity have been proposed, for example,  weak strong convexity \citep{liu2014asynchronous}, restricted secant inequality \citep{zhang2013gradient}, error bound \citep{cannelli2020asynchronous}, quadratic growth \citep{cui2017quadratic}, etc.  
Recently, there has been a surge of interest in analyzing nonconvex functions with block structures. Multiple assumptions have been analyzed, which are correlated to each block  (coordinate) when other blocks are fixed, for example, PL-strongly-concave \citep{guo2020fast}, nonconvex-PL \citep{sanjabi2018solving}, PL-PL \citep{daskalakis2020independent, yang2020global,chen2022faster}, and multi-convex \citep{xu2013block,shen2017disciplined,wang2019multi,wang2022accelerated}. 
For instance, the multi-convexity assumes the convexity of the function with respect to each block when the remaining blocks are fixed.
On the other hand, the other aforementioned conditions are tailored for objective functions comprising only two blocks. They are particularly defined for min-max type optimizations rather than general minimization problems.

The nonconvex optimization realm has seen a growing interest in the gradient dominance condition a.k.a. Polyak-{\L}ojasiewicz (PL) condition. For instance, in analyzing linear quadratic games \citep{fazel2018global}, matrix decomposition \citep{li2018algorithmic}, robust phase retrieval \citep{sun2018geometric} and training neural networks \citep{hardt2016identity,charles2018stability,liu2022loss}.
This is due to its ability to enable sharp convergence analysis of both deterministic GD and stochastic GD algorithms while being satisfied by a wide range of nonconvex functions.
More formally, a function $f$ satisfies the PL condition if there exists a constant $\mu>0$ such that
\begin{equation}\label{strong_PL}
    \|\nabla f(x)\|^2\geq 2\mu\big(f(x)-\min_{y\in\mathbb{R}^d}f(y)\big),\quad \forall x\in\mathbb{R}^d.
\end{equation}

This was first introduced by \citet{polyak1963gradient,lojasiewicz1963topological}, who analyzed the convergence of the GD algorithm under the PL condition and showed its linear convergence to the global minimum. 
This condition can be perceived as a relaxation of strong convexity and as discussed in \citep{karimi2016linear}, it is closely related 
to conditions such as weak-strong convexity \citep{necoara2019linear}, restricted secant inequality \citep{zhang2013gradient} and error bound \citep{luo1993error}. 


As mentioned, the PL condition has been extended and applied to optimization problems with multiple coordinates. 
This extension is analogous to generalizing the concept of convexity (concavity) to convex-concavity. 
For instance, the two-sided PL condition was introduced by \citet{yang2020global} for analyzing deterministic and stochastic alternating gradient descent ascent (AGDA) in \textit{min-max games}. They showed that under this condition, AGDA converges at linear rate to the unique NE. 
It is noteworthy that most literature requires convexity or PL condition to establish the last-iterate convergence rate to the NE \citep{scutari2010convex,sohrabi2020survey,jordan2024adaptive}. 
A considerable relaxation is that the function satisfies the convexity or PL condition when all variables except one are fixed.

\begin{shadowboxtext}
\textit{Can similar results be achieved by extending the two-sided PL condition to accommodate problems in the form of \eqref{min_problem}, where the objectives comprise $n$ coordinates? And is there an algorithm to guarantee convergence to a NE at a linear rate in such problems?}
\end{shadowboxtext}



Motivated by the above question, we introduce the notion of $n$-sided PL\footnote{We should emphasize that $2$-sided PL and two-sided PL are slightly different conditions as the former is suitable for $\min_{x,y}f(x,y)$ while the latter is for $\min_{x}\max_{y}f(x,y)$.} condition (\cref{PL_condition}), which is an extension to the PL condition, and show that it holds in several well-known nonconvex problems such as $n$-player linear quadratic game, linear residual network, etc. 
It is noteworthy that, unlike the two-sided PL condition, which guarantees the uniqueness of the NE in min–max optimization \citep{yang2020global,chen2022faster}, functions satisfying the $n$-sided PL (even 2-side PL) condition may have multiple NE points (see \cref{sec:def} for examples). However, as we will discuss, the set of stationary points for such functions is equivalent to their NE points. 
Moreover, unlike the two-sided PL condition, which ensures linear convergence of the AGDA algorithm to the NE, the BCD algorithm exhibits varying convergence rates for different functions, all satisfying the $n$-sided PL condition. Similar behavior has been observed with multi-convex functions \citep{xu2017globally,wang2019multi}. Therefore, additional local or global conditions are required to characterize the convergence rate under the $n$-sided PL condition. 

In this work, we study the convergence of first-order GD-based algorithms such as the BCD, and propose different variants of BCD that are more suitable for the class of nonconvex functions satisfying $n$-sided PL condition. 
We show the convergence to NE and introduce additional local conditions under which linear convergence can be guaranteed.


%% file: AAMAS-2026-Formatting-Instructions-CCBY/2_related.tex
\subsection{Related Work}

\textbf{Block Coordinate Descent and its variants:} Block coordinate descent (BCD) is an efficient and reliable gradient-based method for optimization problems, which has been used extensively in machine learning \citep{nesterov2012efficiency,allen2016even,zhang2017convergent,zeng2019global,nakamura2021block}. 
 Numerous existing works have studied the convergence of BCD and its variants applied to the special case of potential games. Most of them require the assumptions of convexity, PL condition, and their extensions \citep{beck2013convergence,hong2017iteration,song2021fast,chen2023global,chorobura2023random}. 
For instance, \citet{xu2013block, xu2017globally} studied the convergence of BCD for the regularized block multiconvex optimization. They established the last iterate convergence under Kurdyka-{\L}ojasiewicz, which might not hold for many functions globally. 
The authors in \citep{song2021fast} considered the generalized Minty variational problem and applied cyclic coordinate dual averaging with extrapolation to find its solution. Their algorithm is independent of the dimension of the number of coordinates. However, their results rely on assuming the monotonicity of the operators, which is rarely satisfied in practice.
\citet{cai2023cyclic} considered composite nonconvex optimization and applied cyclic block coordinate descent with PAGE-type variance reduced method. They proved linear and non-asymptotic convergence when the PL condition holds, which is not valid for functions with multiple local minima. However, none of these work consider the general-sum setting.

\textbf{PL condition in optimization:} The PL condition was originally proposed to relax the strong convexity in the minimization problem sufficient for achieving the global convergence for first-order methods. For example, \citet{karimi2016linear} showed that the standard GD algorithm admits a linear convergence to minimize an $L$-smooth and $\mu$-PL function. 
To be specific, in order to find an $\epsilon$-approximate optimal solution $\hat{x}$ such that $f(\hat{x})-f^\star\leq\epsilon$, GD requires the computational complexity of the order $O(\frac{L}{\mu}\log\frac{1}{\epsilon})$.
Besides this, different proposed methods, such as the heavy ball method and its accelerated version, have been analyzed \citep{danilova2020non,wang2022provable}. 
The authors in \citep{yue2023lower} proved the optimality of GD by showing that any first-order method requires at least $\Omega(\frac{L}{\mu}\log\frac{1}{\epsilon})$ gradient costs to find an $\epsilon$-approximation of the optimal solution. Furthermore, many studies focus on the complexity when the objective function has a finite-sum structure, i.e., $f(x)=\frac{1}{n}\sum_{i=1}^{n}f_i(x)$, e.g., \citep{lei2017non, reddi2016stochastic, li2021page, wang2019spiderboost, bai2024complexity}. 
It is important to emphasize that the aforementioned studies do not generalize to optimization problems of the form given in \eqref{min_problem}.

In addition to the minimization problem, extensions of the PL condition, such as two-sided conditions, have been proposed to provide convergence guarantees to saddle points for gradient-based algorithms when addressing minimax optimization problems. 
For example, the two-sided PL holds when both functions $h_y(x):=f(x,y)$ and $h_x(y):=-f(x,y)$ satisfy the PL condition \citep{yang2020global,chen2022faster}, or one-sided PL condition holds when only $h_y(x)$ satisfies the PL condition \citep{guo2020fast,yang2022faster}. 
Various types of first-order methods have been applied to such problems, for example, SPIDER-GDA \citep{chen2022faster}, AGDA \citep{yang2020global}, Multi-step GDA \citep{sanjabi2018solving,nouiehed2019solving}. 
For additional information on the sample complexity of the methods mentioned earlier and their comparisons, see \citep{chen2022faster} and \citep{bai2024complexity}.

%% file: AAMAS-2026-Formatting-Instructions-CCBY/3_definitions.tex
\section{$n$-sided PL Condition}\label{sec:n-pl}

\textbf{Notations:}
Throughout this work, we use $\|\cdot\|$ to denote the Euclidean norm and the lowercase letters to denote a column vector. In particular, we use $x_{-i}$ to denote the vector $x$ without its $i$-th block, where $i\in[n]:=\{1,...,n\}$.
The partial derivative of general function $f(x)$ with respect to the variables in its $i$-th block is denoted as $\nabla_{i}f(x):=\frac{\partial}{\partial x_i}f(x_i,x_{-i})$ and the full gradient is denoted as $\nabla f(x)$ that is $(\nabla_{1}f(x),...,\nabla_{n}f(x))$. 

\subsection{Definitions and Assumptions}\label{sec:def}

In this paper, we assume all the functions $f_i(x):$ $\mathbb{R}^d\to\mathbb{R}$ belong to $C^1$, i.e., they are continuously differentiable. Furthermore, we assume they have Lipschitz gradient.

\begin{assumption}[Smoothness]\label{L_Lip}
We assume the $L$-Lipschitz continuity of the derivative $\nabla f_i(x)$,
\begin{equation*}
    \|\nabla f_i(x)-\nabla f_i(y)\|\leq L_i\|x-y\|,\quad \forall i\in[n], \forall x, y,
\end{equation*}
where $L_i>0$ are constants and $L:=\max_i L_i$. In this case, $f_i(x)$ is also called $L$-smooth.
\end{assumption}


We now define two notions of optimality; Nash equilibrium (NE) and Stationary point. 

\begin{definition}[Nash Equilibrium (NE)]
 Point $x^\star=(x_1^\star,...,x_n^\star)$ is called a Nash Equilibrium of functions $\{f_i(x)\}$ if 
 \begin{equation*}
     f_i(x_{i}^\star,x_{-i}^\star)\leq f_i( x_{i},x_{-i}^\star), \quad \forall i\in[n], \forall x_i\in\mathbb{R}^{d_i}.
 \end{equation*}
The set of all NE points of functions $\{f_i(x)\}$ is denoted by $\mathcal{N}(f_1,\cdots,f_n)$.
\end{definition}

The other notion, the stationary point, is related to the first-order condition of optimality and also relevant for studying gradient-based algorithms.

\begin{definition}[$\varepsilon$-Partial Stationary point]
    Point $\tilde x=(\tilde x_1,...,\tilde x_n)$ is called an $\varepsilon$-stationary point of functions $\{f_i(x)\}$ if 
    \begin{equation*}
    \|\nabla_i f_i(\tilde x)\|\leq \varepsilon, \quad \forall i\in[n].
    \end{equation*}
    When $\varepsilon=0$,  $\tilde x$ is called a partial stationary point. We denote the set of all $\varepsilon$-stationary points  by $\mathcal{S}_\varepsilon(f_1,\cdots,f_n)$, respectively.
\end{definition}

It is important to note that the notion of partial stationarity in general-sum games differs from the standard concept of stationarity commonly used in potential games. In the former, the partial derivative of each player’s objective function with respect to its own variable is zero, whereas in the latter, the entire gradient of the potential function vanishes. Thus, it is possible that, at a partial stationary point $x\in\mathcal{S}_\varepsilon(f_1,\cdots,f_n)$, the full gradient $\nabla f_j(x)\neq0$ for some players $j\in[n]$.

For general nonconvex minimization problems, the above two notions are not necessarily equivalent, i.e., a stationary point may not be a NE. 
Nevertheless, for the remainder of this work, we assume that the objective functions have at least one NE, i.e., $\mathcal{N}(f_1,\cdots,f_n)\neq \emptyset$. We also assume that $\arg\min_{x_i\in\mathbb{R}^{d_i}}f_i(x_i,x_{-i})$ is non-empty for any $i\in[n]$ and $x_{-i}$, i.e., there exists a best response to every $x_{-i}$.
Below, we formally introduce the $n$-sided PL condition for the functions $\{f_i(x)\}$.

\begin{definition}[$n$-sided PL Condition]\label{PL_condition}
We say the set of functions $\{f_i(x)=f_i(x_1,...,x_n): i\in[n]\}$ satisfy $n$-sided $\mu$-PL condition if there exist a set of positive constants $\mu_i>0$ such that for all $x\in\mathbb{R}^{d}, i\in[n]$,
\begin{align}\label{eq:n-sided}
\|\nabla_{i}f_i(x_i,x_{-i})\|^2\geq 2\mu_i\big(f_i(x_i,x_{-i})-f_{i,x_{-i}}^\star\big), 
\end{align}   
where $f_{i,x_{-i}}^\star:=\min_{y_i\in\mathbb R^{d_i}} f_i(y_i,x_{-i})$. We also denote $\mu:=\min_i\mu_i$.
\end{definition}

It is worth noting that the $n$-sided PL condition does not imply convexity or the gradient dominance (PL) condition. It is an extension to the PL condition. In the special case of a potential game where $f_i(x)=f(x)$ for all $i\in[n]$ and $f(x)$ is independent of $x_{-i}$, then $f$ satisfies the PL condition. 
Moreover, it is considerably weaker than multi-strong convexity.
Nevertheless, under the $n$-sided PL condition, the set of stationary points and the NE set are equivalent.

\begin{lemma}\label{SP_equal_NE}
If functions $\{f_i(x_1,...,x_n)\}$ satisfy the $n$-sided PL condition, then $\mathcal{S}_0(f_1,\cdots,f_n)=\mathcal{N}(f_1,\cdots,f_n)$.
\end{lemma} 
All proofs are in the Appendix \ref{sec:proofs}.
It is noteworthy that, unlike the two-sided PL condition, even in the special case of a potential game where all players have the same objective function, under the $n$-sided PL condition presented in \cref{PL_condition}, it is no longer possible to ensure that the NE is unique. 
In fact, there could be multiple NEs with different function values. 
For example, consider the following two potential games with objective functions illustrated in Figure \ref{fig:example1},
\begin{align*}
    & f^{(1)}(x)\!:=\!f_1^{(1)}(x)\!=\!f_2^{(1)}(x)\!=\!(x_1\!-\!1)^2(x_2\!+\!1)^2\!+\!(x_1\!+\!1)^2(x_2\!-\!1)^2\!\!,\\
    &f^{(2)}(x)\!:=\!f^{(2)}_1(x)\!=\!f_2^{(2)}(x)\!=\!f^{(1)}(x)+\exp(-(x_2-1)^2).
\end{align*}
As shown in Appendix \ref{app:example}, both functions $f^{(1)}$ and $f^{(2)}$ are $2$-sided PL, but their sets of NE and the set of minima points are not equivalent. 
In particular, both functions have three NE points, while $f^{(1)}$ has two global minima and a saddle point, and $f^{(2)}$ has a local minima, a global minima, and a saddle point. 
\begin{figure}
     \centering
     \begin{subfigure}
         \centering
         \includegraphics[width=0.45\linewidth,,height=2.5cm]{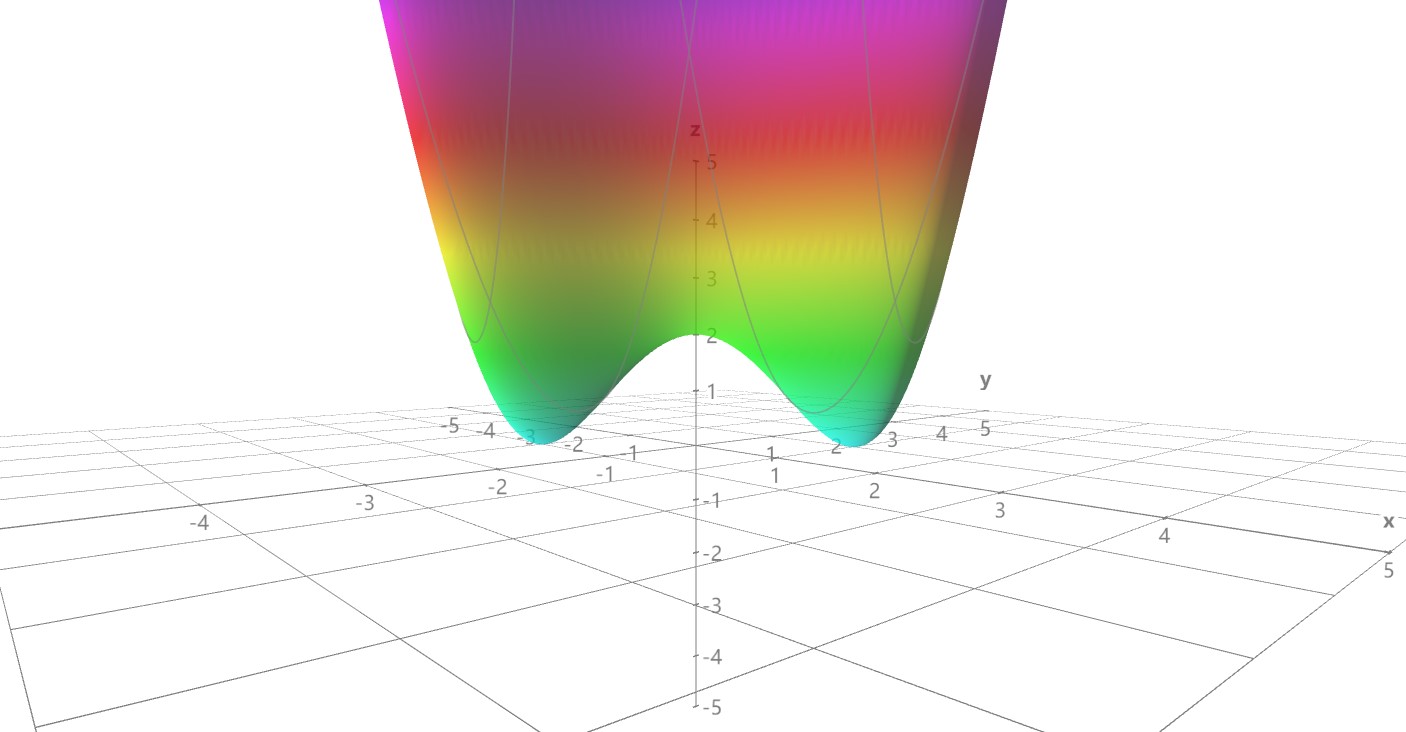}
         \label{fig:y equals x}
     \end{subfigure}
     \begin{subfigure}
         \centering
         \includegraphics[width=0.44\linewidth,,height=2.5cm]{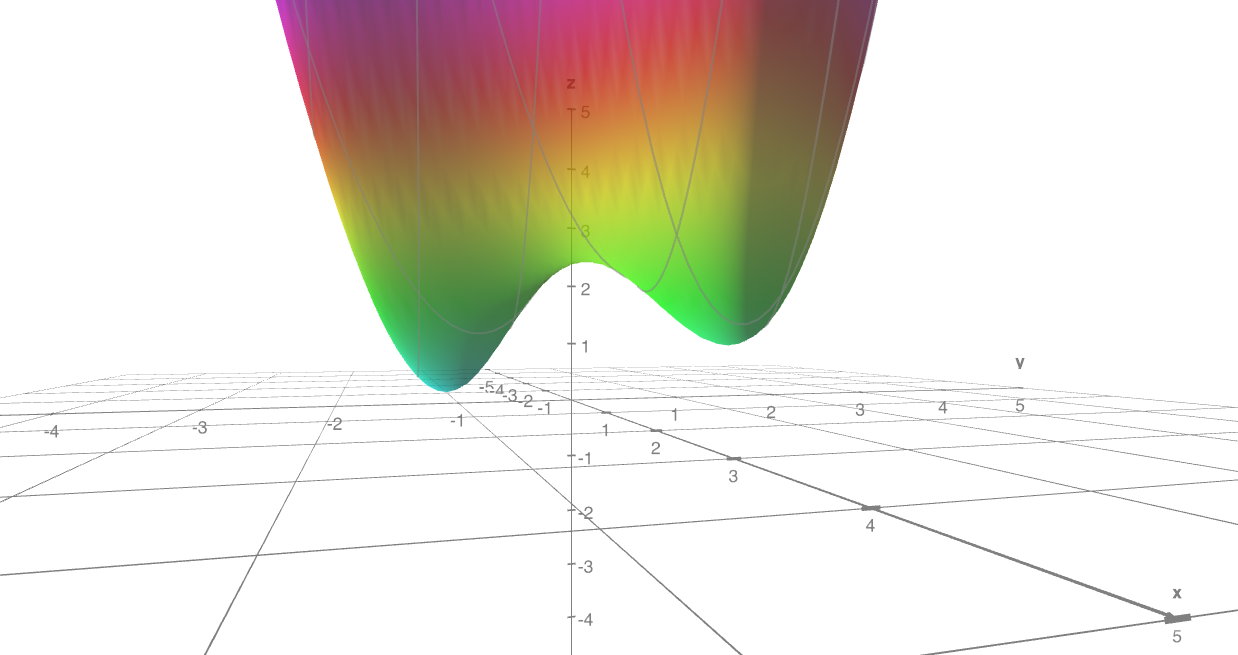}
         \label{fig:three sin x}
     \end{subfigure}
        \caption{Left is function $f^{(1)}(x_1,x_2)$ and right is  $f^{(2)}(x_1,x_2)$.}
        \label{fig:example1}
\end{figure}

%% file: AAMAS-2026-Formatting-Instructions-CCBY/6_algorithms.tex
\section{Algorithms and Convergence Analysis}\label{sec:algs}

In this section, we first study the BCD algorithm for finding a stationary point of \eqref{min_problem} under the $n$-sided PL condition. 
Afterward, we propose different variants of BCD algorithms that can provably achieve better convergence rates. 

The BCD algorithm is a coordinate-wise approach that iteratively improves its current estimate by updating a selected block coordinate using the first-order partial derivatives until it converges. 
Depending on how the coordinates are chosen, various types of BCD algorithms can be devised.
For example, coordinates can be selected uniformly at random for \textit{random BCD} (R-BCD), presented in Algorithm \ref{alg:RBCD} with learning rates $\alpha$, or in a deterministic cyclic sequence for \textit{cyclic BCD}, presented in Algorithm \ref{alg:BCD} in the Appendix. 
Moreover, to update the $i$-th block at the $t$-th iteration, it employs $\nabla_{i}f_i(x^{t-1})$, where $x^{t-1}$ denotes the latest estimated point. 
The next result shows that when the iterates of the BCD are bounded, the output converges to the NE set. 
\begin{algorithm}
\caption{Random Block Coordinate Descent (R-BCD)}\label{alg:RBCD}
\begin{algorithmic}
   \STATE {\bfseries Input:} initial point $x^0=(x_1^0,...,x_n^0)$, learning rates $\alpha$, $t=1$.
   \WHILE{not converges}
   \STATE{Select $i\in[n]$ at random}
   \STATE $x_i^t=x_i^{t-1}-\alpha\nabla_{i}f_i(x^{t-1})$
   \STATE $x_j^t=x_j^{t-1},\quad$ for all  $j\neq i$
   \ENDWHILE
\end{algorithmic}
\end{algorithm}

\begin{figure}
\centering
    \includegraphics[width=8.5cm,height=4.3cm,trim=2.7cm 9.2cm 2.7cm 9.2cm, clip]{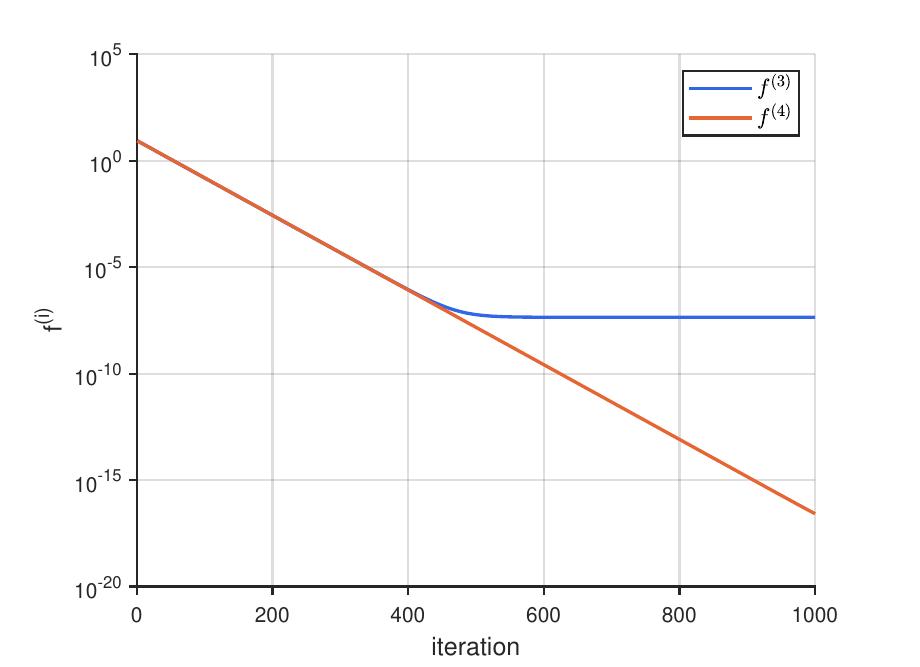}
    \caption{The result of R-BCD algorithm applied to potential setting with functions $f^{(3)}(x_1,x_2)$ and $f^{(4)}(x_1,x_2)$. The y-axis is in log scale, thus the R-BCD demonstrates linear convergence for $f^{(4)}$.
   }\label{bcd_low}
\end{figure}
Note that AGDA can be viewed as a variant of BCD applied to a two-player potential game where the first coordinate is updated via descent and the second coordinate via ascent.
Unlike the two-sided PL condition that leads to linear convergence of AGDA to the min-max, the $n$-sided PL condition does not necessarily lead to any specific convergence rate of the BCD. 
To illustrate this phenomenon, we consider two settings: potential and general-sum. In each setting, we examine two different games and analyze the convergence behavior of R-BCD. 

\paragraph{Potential setting:} We consider the following two potential games, i.e., $f_1=f_2=f$, in which all players share the same objective function; both games satisfy the $2$-sided PL condition.
\begin{align*}
    &f^{(3)}(x_1,x_2)=\begin{cases} 
      (x_1+x_2)^2+\exp\big(-\frac{1}{(x_1-x_2)^2}\big), \ \  x_1\neq x_2, \\
      (x_1+x_2)^2, \ \  \text{o.t.}, \\
   \end{cases}\\
   &f^{(4)}(x_1,x_2)=(x_1+x_2)^2.
\end{align*}
We applied the BCD algorithm to both these functions with constant learning rates to find their NE points with random initializations. 
As it is illustrated in Figure \ref{bcd_low}, the R-BCD converges linearly for the function $f^{(4)}$ while it converges sub-linearly for $f^{(3)}$.

\paragraph{General-sum setting:}
Here, we consider the following two-player games, both of which satisfy the $2$-sided PL condition. The objectives of the first game are
\begin{align*}
    &f^{(5)}_{1}(x_1,x_2)=\begin{cases} 
      (x_1+x_2)^2+\exp\big(-\frac{1}{(x_1-x_2)^2}\big), \ \  x_1\neq x_2, \\
      (x_1+x_2)^2, \ \  \text{o.t.}, \\
   \end{cases}\\
   &f^{(5)}_{2}(x_1,x_2)=(x_1+x_2)^2,
\end{align*}
and the objectives of the second game are
\begin{align*}
    &f^{(6)}_{1}(x_1,x_2)=x_1^2+x_2^2,\quad f^{(6)}_{2}(x_1,x_2)=(x_1+x_2)^2.
\end{align*}
Figure \ref{fig:g_example} illustrates the convergence rate of R-BCD applied to the above functions, showing clearly that the convergence rates differ between them. 
Note that we plotted the sum of the players’ objectives against the number of iterations, since in both examples $(0,0)$ corresponds to the global Nash equilibrium, where all players’ objectives are zero.

\begin{figure}
\centering
    \includegraphics[width=8.5cm,height=4.2cm,trim=2.7cm 9.2cm 2.7cm 9.2cm, clip]{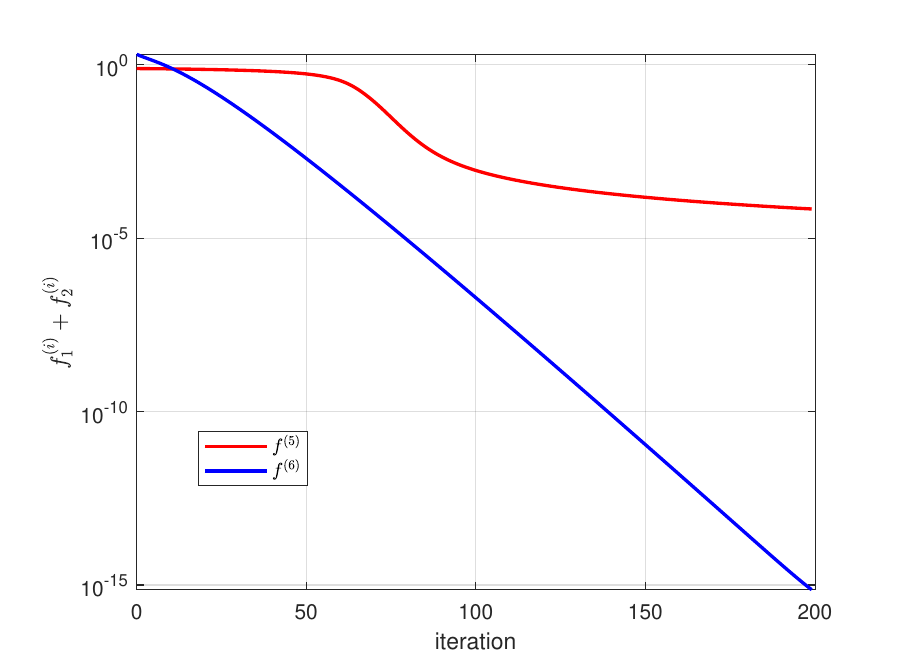}
    \caption{The result of R-BCD algorithm applied to functions $\{f^{(5)}_{1},f^{(5)}_{2}\}$ and $\{f^{(6)}_{1},f^{(6)}_{2}\}$. The y-axis is in log scale, thus the BCD shows linear convergence in the second game.}\label{fig:g_example}
\end{figure}

These examples illustrate that, even under the $n$-sided PL condition and smoothness, characterizing the convergence rate of the BCD algorithm in general-sum games may not be feasible and further assumptions on the function class are required.
Next, we study one such assumption that holds for a large class of non-convex functions and could characterize the convergence rate of R-BCD.

\subsection{ Convergence of random BCD under an additional assumption} 
To introduce our additional assumption, we need to define the sum of the objective functions and their best responses, denoted by $F(x)$ and $G_F(x)$, respectively, which plays a central role in analyzing the convergence of coordinate-wise algorithms. 
\begin{align}\label{eq:G_Function}
&F(x):=\sum_{i=1}^n f_i(x),\quad F_{-i}(x):=F(x)-f_i(x),\\
   & G_F(x):=\sum_{i=1}^n f_i(x^*_i(x),x_{-i}),
\end{align}
where $x^*_i(x)$ denotes the best response to $x_{-i}$, i.e., 
$$
x^*_i(x)\in\arg\min_{y_i}\big\{\|y_i-x_i\|:\ f_i(y_i,x_{-i})\leq f_i(z_i,x_{-i}),\forall z_i\in\mathbb R^{d_i}\big\}.
$$

Next result shows that the best response function is smooth when the objective functions are both smooth and $n$-sided PL. 

\begin{lemma}\label{G_smooth}\label{F-G_equivalence}
   If $\{f_i(x)\}$ satisfy $n$-sided $\mu$-PL and \cref{L_Lip}, then $G_F(x)$ is $nL'$-smooth, where $L':=L+\frac{L^2}{\mu}$.
\end{lemma}

It is straightforward to see that $F(x)-G_F(x)\geq 0$ for all $x$.
Moreover, if $x^*\in\mathcal{N}(f_1,\cdots,f_n)$, then $F(x^*)=G_F(x^*)$. Conversely, if $F(x^*)-G_F(x^*)=0$, then $f_i(x^\star)=\min_{x_i}f_i(x_i,x_{-i}^\star)$, $\forall i$, which implies that $x^\star$ belongs to $\mathcal{N}(f_1,\cdots,f_n)$. The next lemma summarizes these results. 
 
\begin{lemma}\label{F-G_Nash}
    $x^\star$ is a NE if $F(x^\star)-G_F(x^\star)=0$. Moreover, $x\in\mathcal{S}_\epsilon$ if $F(x)-G_F(x)\leq\sqrt{2L\epsilon}$.
\end{lemma}

Thus, $F(x) - G_F(x)$ serves as an indicator of whether $x$ is a NE. Using this indicator, we show that the R-BCD algorithm achieves a linear convergence rate under an additional assumption.

\begin{shadowboxtext}
\begin{theorem}\label{convergence_ideal_kt0}
Suppose $\{f_i(x)\}$ are $n$-sided $\mu$-PL functions satisfying \cref{L_Lip} such that for all $x\in\mathbb R^d$
\begin{equation}\label{eq:add_ass}
\sum_{i=1}^n\big\langle\nabla_{i} \big(G_F(x)-F_{-i}(x)\big),\nabla_{i} f_{i}(x)\big\rangle\leq \kappa\sum_{i=1}^n \|\nabla_i f_i(x)\|^2,    
\end{equation}
where $\kappa<1$, then random BCD with 
$\alpha\leq\frac{1-\kappa}{n(L+L')}$
achieves linear convergence rate, i.e., 
\begin{equation*}
    \mathbb{E}[F(x^{t+1})-G_F(x^{t+1})] \leq \Big(1-\frac{\mu\alpha(1-\kappa)}{2}\Big)\ \mathbb{E}[F(x^t)-G_F(x^t)].
\end{equation*}

The expectation is taken over the randomness inherent in the procedure for selecting coordinates.
\end{theorem}
\end{shadowboxtext}


 It is straightforward to verify that  $f^{(4)}$ in Figure~\ref{bcd_low} (for all $x \in \mathbb{R}^2$), $f^{(6)}$ in Figure~\ref{fig:g_example} (for all $x \in \mathbb{R}^2$), and  $f^{(0)}$ in Figure~\ref{fig:example1} (for all points in $\{(x_1, x_2) : |x_1| > 0.75, |x_2| > 0.75\}$) satisfy the condition in \eqref{eq:add_ass}. Notably, these sets contain all the NE points, and the R-BCD algorithm exhibits a linear convergence rate when applied to these functions. The above result further highlights the fast convergence of R-BCD for finding NEs of these functions.

To demonstrate that the additional assumption in \eqref{eq:add_ass} is not overly restrictive, we apply the Cauchy–Schwarz inequality together with the result of Lemma \ref{G_deri_bound} (see Appendix) to obtain the following bound, valid for all smooth functions $\{f_i\}$ satisfying the $n$-sided PL condition:
\begin{align*}
\sum_{i=1}^n \big\langle \nabla_i G_F(x) - \nabla_i F_{-i}(x), \nabla_i f_i(x) \big\rangle
\le \big(\sqrt{3n}\tfrac{L}{\mu}\big) \sum_{i=1}^n \|\nabla_i f_i(x)\|^2,
\end{align*}
where $\sqrt{3n}\tfrac{L}{\mu} \ge 1$ (see Lemma \ref{f_G_nabla} for a proof that $\mu \leq L$).
This implies that for any collection of smooth functions satisfying the $n$-sided PL condition, there exists a constant $c$ such that the left-hand side of\eqref{eq:add_ass} is bounded by $c\sum_i \|\nabla_i f_i(x)\|^2$. 
According to Theorem \ref{convergence_ideal_kt0}, R-BCD exhibits linear convergence to a NE when the corresponding constant factor is less than one. As discussed in the previous section, however, there also exist smooth, $n$-sided PL functions (e.g., those depicted in Figures \ref{bcd_low} and \ref{fig:g_example}) whose constant factors exceed one but remain bounded by $\sqrt{3n}\tfrac{L}{\mu}$. 

\subsection{Adaptive random BCD with the knowledge of exact best responses}

As discussed earlier, BCD algorithms may exhibit different convergence rates in general. In the remainder of this section, we therefore focus on developing variants of the randomized BCD algorithm that can achieve convergence rates close to linear.
To this end, we propose algorithms that, at each iteration $t$, update the variables based on the best responses $\{x^*_i(x^t)\}$.

We initially propose an algorithm that presumes access to the exact values of the best responses at each iteration. 
Subsequently, we refine this assumption by integrating a sub-routine into the proposed algorithm capable of approximating the best responses.
To present our results, we need the following definition.

\begin{definition} [$(\theta,\nu)$-PL condition]
A function $f$ with $\min_x f(x)=0$ satisfies $(\theta,\nu)$-PL condition if and only if there exists $\theta \in[1,2)$ and $\nu>0$ such that $\|\nabla f(x)\|^\theta \geq(2 \nu)^{\theta / 2} f(x). \quad$    
\end{definition}

It has been proved by \citet{lojasiewicz1963topological} that for any $C^1$ analytic function, there exists a neighborhood $U$ around the minimizer where $(\theta,\nu)$-PL condition is satisfied for some $\theta$ and $\nu$.

\begin{algorithm}
   \caption{Ideal Adaptive random BCD (IA-RBCD)}
\begin{algorithmic}\label{arbcd_ideal}
   \STATE {\bfseries Input:} initial point $x^0=(x_1^0,...,x_n^0)$, $T$, step size $\alpha$, $0\leq\gamma<1$ and $C>0$
   \FOR{$t=0$ {\bfseries to} $T-1$}
   \STATE $A=\sum_{i=1}^n\big\langle\nabla_{i} \big(G_F(x^t)-F_{-i}(x^t)\big),\nabla_{i} f_{i}(x^t)\big\rangle$ 
   \STATE $B=\sum_{i=1}^n\|\nabla_{i} (G_F(x^t)-F_{-i}(x^t))\|^2$
   \STATE $D=\sum_{i=1}^n \|\nabla_i f_i(x^t)\|^2$
   \STATE sample $i^t$ uniformly at random from $\{1,2,...,n\}$
   \IF{ $A \leq \gamma D$}
   \STATE $k^t=0$ \hfill\texttt{:Case 1:}
   \ELSIF{$(B-A)^2\geq C A^2$}
   \STATE $k^t=-2+\frac{A}{B}$ \hfill\texttt{:Case 2:}
   \ELSE   \STATE $k^t=-1$ \hfill\texttt{:Case 3:}
   \ENDIF
   \STATE $x_{i^t}^{t+1}=x_{i^t}^{t}-\alpha\Big(\nabla_{i^t} f_{i^t}(x^t)+k^t\nabla_{i^t} (G_F(x^t)-F_{-i^t}(x^t))\Big),$ 
   \STATE $x_{i}^{t+1}=x_{i}^{t}$\quad for all\quad $i\neq i^t$.
   \ENDFOR
\end{algorithmic}
\end{algorithm}

Algorithm \ref{arbcd_ideal} presents the steps of our modified version of the random BCD. In this algorithm, instead of updating along the direction of $-\nabla_{i^t} f_{i^t}(x)$, where $i^t$ denotes the chosen coordinate at iteration $t$, a linear combination of $\nabla_{{i^t}}f_{i^t}(x)$ and $\nabla_{i^t}\big(G_F(x)-F_{-i^t}(x)\big)$ is used to refine the updating directions.
The coefficient of this linear combination, \(k^t\), is adaptively selected based on the current estimated point. 
It is important to mention that $\nabla G_F(x)$ can be computed using the gradients of $\{f_i\}$ at the best responses,
\begin{align}\label{eq:gradient_g}
    \nabla G_F(x)=\sum_{i=1}^n\nabla f_i\big(x^*_i(x), x_{-i}\big).
\end{align}


The following theorem establishes the convergence guarantee of the algorithm presented in \ref{arbcd_ideal}.
\begin{shadowboxtext}
\begin{theorem}\label{convergence_ideal}
For $n$-sided $\mu$-PL functions $\{f_i(x)\}$ satisfying \cref{L_Lip}, by applying \cref{arbcd_ideal}, 

\begin{itemize}[leftmargin=8pt]
\item in Case 1 with $\alpha\!\leq\!\frac{1-\gamma}{n(L+L')}$, we have 
\begin{equation*}
\mathbb{E}[F(x^{t+1})\!-\!G_F(x^{t+1})|x^t]\! \leq\! \Big(1\!-\!\frac{\mu\alpha(1-\gamma)}{2}\Big)\big(F(x^t)\!-\!G_F(x^t)\big),
\end{equation*}

\item in Case 2 with $\alpha\!\leq\min\{\frac{1}{2n(L+L')},\frac{C}{2n(L+L')}\}$, we have
\begin{equation*}
\mathbb{E}[F(x^{t+1})\!-\!G_F(x^{t+1})|x^t]\!\leq\!\! \Big(1\!-\!\!\frac{(L+L')\mu\alpha^2}{2}\Big)\big(F(x^t)\!-\!G_F(x^t)\big),
\end{equation*}

\item in Case 3 with $\alpha\!\leq\!\frac{1}{n(L+L')}$, $\mathbb{E}[F(x^{t+1})\!-\!G_F(x^{t+1})]$ is non-increasing. Furthermore, if $F-G_F$ satisfies $(\theta,\nu)$-PL condition and case 3 are satisfied from iterates $t$ to $t+k$, we have
\begin{equation*}
    \mathbb{E}[F(x^{t+k})-G_F(x^{t+k})|x^t]=\mathcal{O}\Big(\frac{F(x^{t})-G_F(x^{t})}{k^{\frac{\theta}{2-\theta}}}\Big).    
\end{equation*}
\end{itemize}
\end{theorem}
\end{shadowboxtext}

The proof in the Appendix provides the exact constant factors.
It is noteworthy to highlight that BCD requires the objective function to be lower bounded to converge to the NE \citep{xu2013block} and almost surely avoids strict saddle points \citep{lee2016gradient}. However, the above result shows that IA-RBCD converges to the NE, irrespective of this assumption. 
Additionally, based on this result, IA-RBCD demonstrates a linear rate for two out of three cases and a sub-linear rate for the third case. Next, we show that even if the third case occurs frequently, but within a limited range, linear convergence is still guaranteed by IA-RBCD. 

\begin{theorem}\label{case_many_enough}
    Under the assumptions of \cref{convergence_ideal}, by applying Algorithm \ref{arbcd_ideal} for $t$ iterations and denoting the number of iterations it visits Case 3 by $b(t)$, if
    $ B:=\lim\sup_{t\to +\infty} b(t)/t<1$,
    then there exists $0<c_B<1$ and $\hat{T}_B$ such that for all $t\geq \hat{T}_B$,
    \begin{equation*}
        \mathbb{E}[F(x^t)-G_F(x^t)]\leq (c_B)^{t}\big(F(x^0)-G_F(x^0)\big).
    \end{equation*}
\end{theorem}

In the following theorem, we demonstrate that as $\gamma \to 1$ and $C \to 0$, the measure of non-NE points that fail to satisfy case 1 and case 2 converges to zero. Furthermore, when the iterates fall into case 3, $F-G_F$ is small at those iterates which indicates that a good approximation of a NE is already achieved. 
As a consequence, by choosing hyperparameters $\gamma$ and $C$ close to one and zero, respectively, we can achieve close to linear convergence rate.

\begin{theorem}\label{thm:measure}
    Under the assumption of \cref{convergence_ideal} and when $F-G_F$ satisfies $(\theta, \nu)$-PL condition, let $S(\gamma, C)$ be the set of non-NE points that do not satisfy case 1 and case 2, then
    \begin{equation*}
        \lim_{\gamma\to 1, C\to 0} |S(\gamma,C)|=0,
    \end{equation*}
    where $|S(\gamma, C)|$ denotes the measure of the set $S(\gamma, C)$. Moreover, if $S(\gamma,C)$ is non-empty, then
        $$
        \lim_{\gamma\to 1, C\to 0} \max_{x\in S(\gamma, C)}F(x)-G_F(x)=0.
        $$
\end{theorem}

\subsection{Adaptive random BCD without the knowledge of exact best responses}

Evaluating $G_F$ at a given point requires knowledge of the best responses at that point. Often, these best responses are not known a priori, and they have to be computed at each iteration.  
Fortunately, since in our study, $\{f_i\}$ satisfy the $n$-sided PL condition, the best responses can be efficiently approximated by applying the GD algorithm with the partial gradients as a sub-routine. 
Algorithm \ref{sub-routine} outlines the steps of the sub-routine, which outputs an approximation of the best responses for a given step size $\beta$ after $T'$ iterations. These approximated best responses are then used by Algorithm \ref{arbcd_practical} to find a NE point.

\begin{algorithm}
   \caption{Approximating Best Responses (ABR)}\label{sub-routine}
\begin{algorithmic}
   \STATE {\bfseries Input:} Point $x=(x_1,...,x_n)$, $T'$, step size $\beta$
   \FOR{$j=1,...,n$}
   \STATE $y_j^{0}=x_{j}$
   \FOR{$\tau=0,...,T'-1$}
   \STATE $y_j^{\tau+1}=y_j^{\tau}-\beta\nabla_{j}f_j(y_j^{\tau},x_{-j})$
   \ENDFOR
   \ENDFOR
   \STATE {\bfseries Output:} $\nabla\tilde G_F(x)=\sum_{i=1}^n\nabla f_i(y_i^{T'},x_{-i})$
\end{algorithmic}
\end{algorithm}

\begin{lemma}
    For $n$-sided $\mu$-PL functions $\{f_i(x)\}$ satisfying \cref{L_Lip}, by implementing \cref{sub-routine}, with $\beta\leq\frac{1}{L}$ and $T'\geq\log\big(\frac{nL^2}{\mu^2\delta}\big)/\log(\frac{1}{1-\mu\beta})$, 
for any $\delta>0$, we have
\begin{equation*}
    \|\nabla G_F(x)-\nabla\tilde G_F(x)\|^2\leq\delta\sum_{i=1}^n\|\nabla_i f_i(x)\|^2.
\end{equation*}
\end{lemma}

Interestingly as we showed in the next result, the number of steps for approximating $G_F$, $T'$, only depends on the function parameters, and it is independent of the final precision of $F-G_F$. The exact form of the step sizes are presented in the Appendix \ref{app:p_convergence_practical}.

\begin{algorithm}
   \caption{Adaptive random BCD (A-RBCD)}
\begin{algorithmic}\label{arbcd_practical}
   \STATE {\bfseries Input:} initial point $x^0=(x_1^0,...,x_n^0)$, $T, T'$,  step sizes $\alpha$, $\beta$, $0<\gamma<1$ and $C>0$
   \FOR{$t=0$ {\bfseries to} $T-1$}
   \STATE sample $i^t$ uniformly from $\{1,2,...,n\}$
   \STATE $\nabla\tilde G_F(x^t)=$ ABR$(x^{t}, \beta, T')$ \hfill \texttt{:Algorithm \ref{sub-routine}}
   \STATE $\tilde A=\sum_{i=1}^n\big\langle\nabla_{i} \big(\tilde G_F(x^t)-F_{-i}(x^t)\big),\nabla_{i} f_{i}(x^t)\big\rangle$ 
   \STATE $\tilde B=\sum_{i=1}^n\|\nabla_{i} (\tilde G_F(x^t)-F_{-i}(x^t))\|^2$
   \STATE $\tilde D=\sum_{i=1}^n \|\nabla_i f_i(x^t)\|^2$
   \IF{$\tilde A\leq \gamma\tilde D$}
   \STATE $\tilde k^t=0$ \hfill\texttt{:Case 1:}
   \ELSIF{$(\tilde B-\tilde A)^2\geq 2C\tilde A^2$ }
   \STATE $\tilde k^t=-2+\frac{\tilde A}{\tilde B}$\hfill\texttt{:Case 2:}
   \ELSE
   \STATE $\tilde k^t=-1$ \hfill\texttt{:Case 3:}
   \ENDIF
   \STATE $x_{i^t}^{t+1}=x_{i^t}^{t}-\alpha\Big(\nabla_{{i^t}}f_{i^t}(x^{t})+\tilde k^t\nabla_{i^t}\big(\tilde G_F(x^t)- F_{-i^t}(x^t)\big)\Big)$
   \STATE $x_{i}^{t+1}=x_{i}^{t}$,\quad  if \quad $i\neq i^t$
   \ENDFOR
\end{algorithmic}
\end{algorithm}

\begin{shadowboxtext}
\begin{theorem}\label{convergence_practical}
For $n$-sided $\mu$-PL functions $\{f_i(x)\}$ satisfying \cref{L_Lip}, by implementing \cref{arbcd_practical} with $\alpha$ small enough, $\beta\leq\frac{1}{L}$ and $T'\geq C'\log\big(\frac{1}{\alpha}\big)/\log(\frac{1}{1-\mu\beta})$ where $C'$ only depends on the function parameters, 

\begin{itemize}[leftmargin=7pt]
\item in Case 1, we have
\begin{equation*}
\mathbb{E}[F(x^{t+1})-G_F(x^{t+1})|x^t]\leq \Big(1-\frac{\mu\alpha(1-\gamma)}{2}\Big)(F(x^t)-G_F(x^t)),
\end{equation*}

\item in Case 2, we have
\begin{equation*}
    \mathbb{E}[F(x^{t+1})-G_F(x^{t+1})|x^t]\leq \Big(1-\frac{(L+L')\mu\alpha^2}{4}\Big)\ (F(x^t)-G_F(x^t)),
\end{equation*}

\item in Case 3, $\mathbb{E}[F(x^{t})-G_F(x^{t})]$ is non-increasing. Furthermore, if $F\!-\!G_F$ satisfies $(\theta,\nu)$-PL condition and case 3 occurs from iterates $t$ to $t+k$, then
\begin{equation*}
    \mathbb{E}[F(x^{t+k})-G_F(x^{t+k})|x^t]=\mathcal{O}\Big(\frac{F(x^{t})-G_F(x^{t})}{k^{\frac{\theta}{2-\theta}}}\Big). 
\end{equation*}
\end{itemize}
\end{theorem}
\end{shadowboxtext}

%% file: AAMAS-2026-Formatting-Instructions-CCBY/7_applications.tex
\section{Applications}\label{sec:applications}

In this section and in Appendix \ref{sec:add_app}, we evaluate the performance of our algorithms on a set of well-known nonconvex problems that satisfy the $n$-sided PL condition. 

\paragraph{\textbf{Benchmarks}}
We compared our algorithms against the following state-of-the-art algorithms: \textbf{BM1}, which corresponds to the R-BCD algorithm described in \ref{alg:RBCD} and \textbf{BM2}, an empirical variant of Algorithm A-RBCD in which updates are performed exclusively via case 3 while ignoring the first two cases, i.e., $\tilde{k}^{t} = -1$.


\begin{figure}
\centering
\includegraphics[width=8.45cm,height=4.2cm,trim=2.7cm 9.2cm 2.7cm 9.2cm, clip]{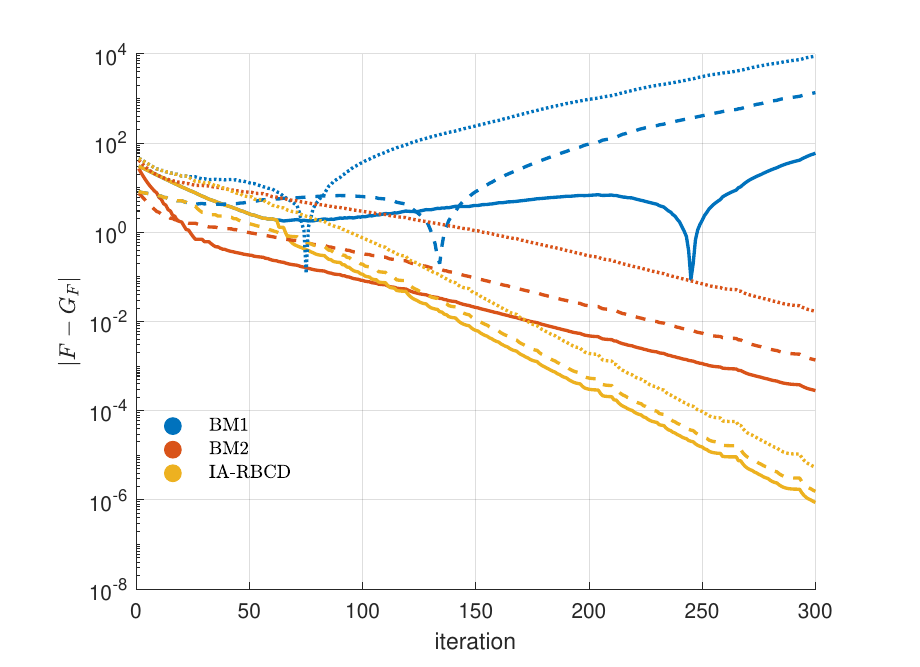}
\caption{ Performances of A-RBCD, BM1, and BM2 applied to functions $f_1$ and $f_2$.}\label{fig:non_convergence}
\end{figure}

\paragraph{\textbf{Function with only strict saddle point:}} 
To illustrate an example where our A-RBCD algorithm converges to a NE point while the random BCD (BM1) fails, we considered a two-player general-sum game with the objective functions $f_1(x_1,x_2)=(x_1-1)^2+4(x_1+0.1\cos(x_1))x_2+(x_2+0.1\sin(x_2))^2$ and $f_2(x_1,x_2)=(x_1-1)^2+4(x_1-0.1\cos(x_1))x_2+(x_2-0.1\sin(x_2))^2$. The problem aims at finding the NE $(x_1^\star,x_2^\star)$, i.e.,  
\begin{equation*}
\begin{aligned}
    f_1(x_1^\star,x_2^\star)\leq f_1(x_1,x_2^\star),\quad     f_2(x_1^\star,x_2^\star)\leq f_2(x_1^\star,x_2),\quad \forall x_1, x_2.    
\end{aligned}
\end{equation*}

Figure \ref{fig:non_convergence} illustrates the convergence results with different initializations. The iterates of A-RBCD always converge to the NE at a linear rate. Note that the NE is a strict saddle point.

\begin{figure*}
\centering
\subfigure[]{
\includegraphics[width=0.45\linewidth,height=4cm,trim=2.7cm 9.2cm 2.7cm 9.2cm, clip]{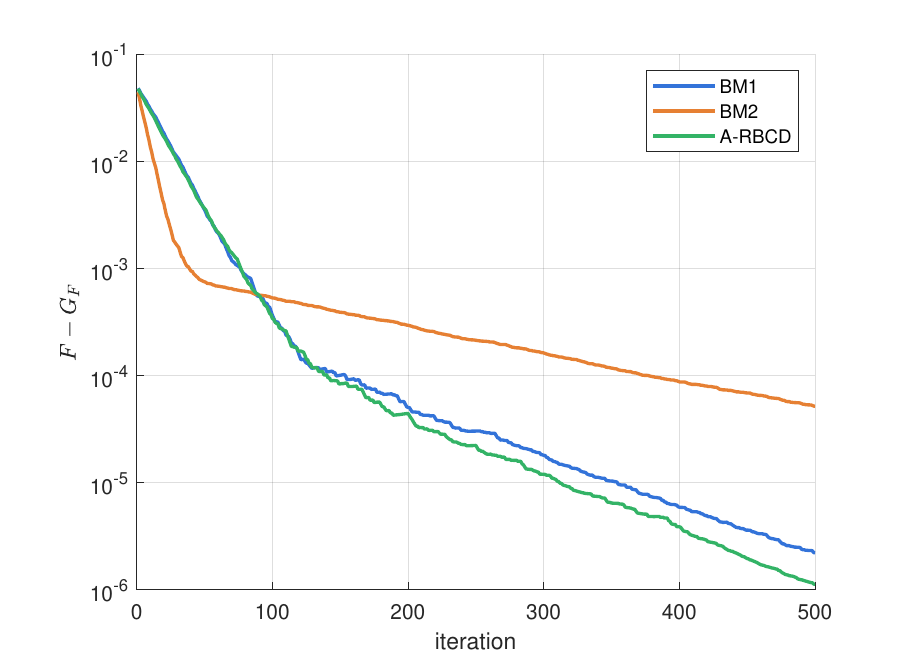}}
\subfigure[]{
\includegraphics[width=0.45\linewidth,height=4cm,trim=2.7cm 9.2cm 2.7cm 9.2cm, clip]{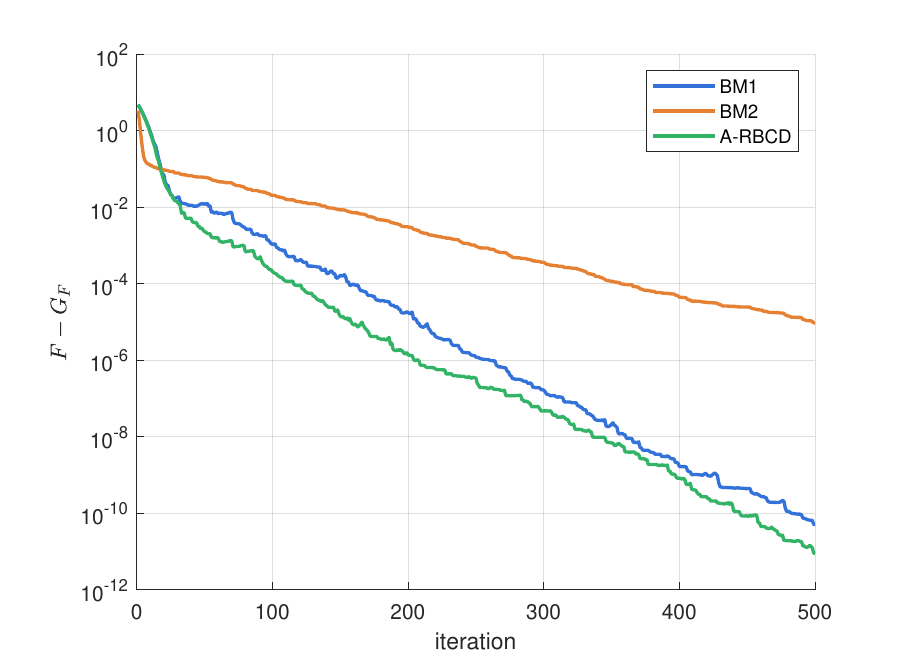}}
\caption{(a) Performances of A-RBCD, BM1, and BM2 applied to the Cournot model with $P(Q)=a-bQ$; (b) with $P(Q)=a-bQ^2$. }\label{fig:compare_two}
\end{figure*}

\paragraph{\textbf{Cournot Competition:}}
The Cournot competition model, first proposed by Augustin Cournot in 1838,
describes a market structure in which multiple firms compete simultaneously by choosing quantities of a homogeneous good \cite{allaz1993cournot,maskin1987theory,abolhassani2014network}. Each firm $i$ determines its output level $q_i$ assuming that the quantities of all other firms $q_{-i}$ are fixed, and the market price is determined by an inverse demand function denoted by $P(Q)$ that depends on the aggregate output $Q=\sum_i q_i$. 
The strategic interdependence among firms is captured by the dependence of each firm's profit $f_i(q_i,q_{-i}) :=-( P(Q)q_i - C_i(q_i))$ on the total market output, where $C_i(q_i)$ is the cost of producing the $q_i$ output. A \emph{Cournot-NE} is a fixed point in which every firm's output $q_i^\star$ is the best response to the others, i.e.,
$$
  f_i(q_i^\star,q_{-i}^\star) 
  \leq f_i(q_i,q_{-i}^\star), 
  \quad \forall\, q_i \ge 0.
$$
Normally, the equilibrium cannot be solved in closed form, and each firm's output decreases with its marginal cost while increasing with market demand. 
The Cournot framework has become a cornerstone in industrial organization \cite{perakis2014efficiency}, energy economics \cite{zhang2020cournot}, and networked systems \cite{bimpikis2019cournot}, providing a tractable yet insightful model for analyzing strategic interactions among self-interested agents under cost or capacity constraints.

In \cref{fig:compare_two}, we compared the convergence of A-RBCD and two benchmark methods for computing the Cournot-NE under two types of inverse demand functions: linear  $P(Q)=a-bQ$ and nonlinear $P(Q)=a-bQ^2$.
The A-RBCD algorithm demonstrates faster and more stable convergence than the benchmark methods. This result highlights that the adaptive correction term in A-RBCD effectively combines $\nabla_i f_i$ and $\nabla_i(G_F - F_{-i})$, leading to improved convergence behavior.

\paragraph{\textbf{Infinite Horizon $n$-player Linear-quadratic (LQ) Game:}} LQ games provide a powerful tool for analyzing the behavior of the multi-agent reinforcement learning (MARL) with continuous state and action space. Owing to their closed-form solvability and interpretability, LQ games serve as canonical benchmarks in both control theory and game-theoretic learning, bridging the gap between theoretical analysis and practical algorithm design \cite{bu2019global,zhang2019policy,zhang2021derivative}.  

In this game, the $i$-th player's objective function is given by
\begin{equation*}
    \mathbb{E}_{s_0\sim\mathcal{D}}\left[\sum_{t=0}^{+\infty}[(s_t)^T Q_i s_t+ (u_t^i)^T R_i u_t^i]\right],
\end{equation*}
where $s_t$ denotes the state, $u_t^i$ is the input of $i$-th player at time $t$, and $i\in[n]$. The state transition of the system is characterized by 
$s_{t+1}=As_t+\sum_{i=1}^n B_iu_t^i+w_t$, where $A\in\mathbb R^{d\times d}$ and $B\in\mathbb R^{d\times k_i}$.
When players apply a linear feedback strategy, i.e., $u_t^i=-K_i s_t$, the objective functions become
\begin{equation*}
    f_i(K_i,K_{-i})=\mathbb{E}_{s_0\sim\mathcal{D}}\left[\sum_{t=0}^{+\infty}[(s_t)^T Q_i s_t+ (K_is_t)^T R_i K_is_t]\right].
\end{equation*}

\begin{figure*}
\centering
\subfigure[\footnotesize{Convergence when $n=3$,\ $d=2$}]{
\includegraphics[width=0.45\linewidth,height=4cm,trim=2.7cm 9.2cm 2.7cm 9.2cm, clip]{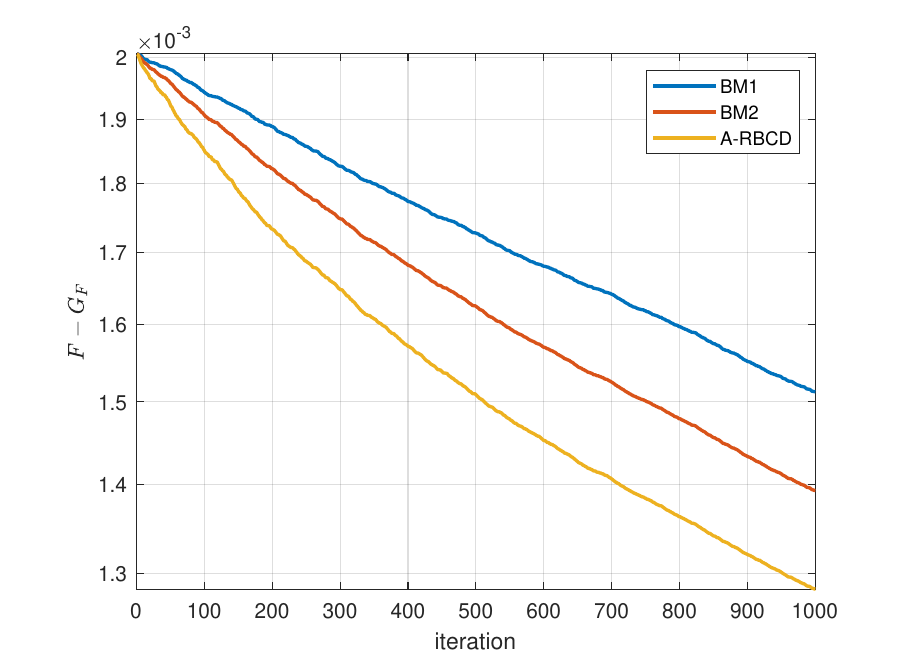}}
\subfigure[\footnotesize{Convergence when $n=5$,\ $d=3$}]{
\includegraphics[width=0.45\linewidth,height=4cm,trim=2.7cm 9.2cm 2.7cm 9.2cm, clip]{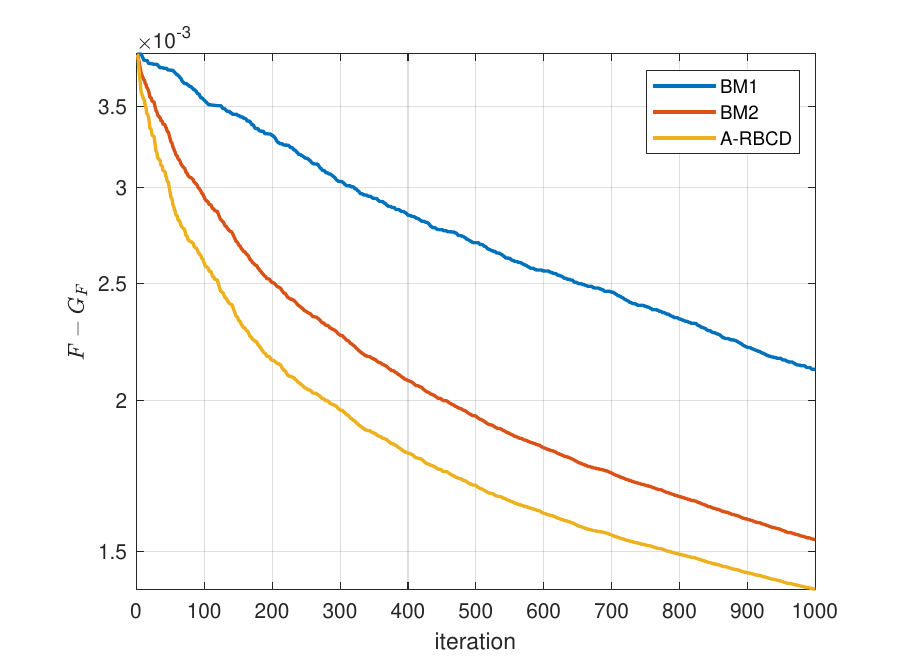}}
\caption{(a) and (b) Performance of A-RBCD, BM1, and BM2 applied to the LQR game for a different set of parameters. The y-axis is on a log scale.}\label{fig:lqr}
\end{figure*}
If $K_i$s are bounded and $\Sigma_0=\mathbb{E}_{s_0\sim\mathcal{D}}[s_0(s_0)^T]$ is full rank, the objective functions $\{f_i\}$ satisfy the $n$-sided PL condition (see  \cref{proof_n_sided_LQR} for a proof). 
However, as it is discussed in \cite{fazel2018global}, even the objective of the one-player LQ game is not convex. Subsequently, the objective function of the $n$-player LQ game is not multi-convex. See \cref{ce_LQR_convex} for examples. 

Figure \ref{fig:lqr} demonstrates the performances of A-RBCD, BM1, and BM2 for different set of parameters.
Note that the benchmark (BM2) reflects the algorithmic setting studied in \cite{hambly2023policy}, which establishes the convergence of both policy gradient descent and vanilla gradient descent in the context of $n$-player general-sum LQ games. Nevertheless, the theoretical guarantee provided in \cite{hambly2023policy} critically depends on the assumption that the covariance of the injected noise is sufficiently large to stabilize the dynamics and it is limited to the class of LQ games which is a special case of our setting.

%% file: AAMAS-2026-Formatting-Instructions-CCBY/8_conclusion.tex
\section{Conclusion}

In this paper, we identified a subclass of nonconvex functions called $n$-sided PL functions and studied the convergence of GD-based algorithms, particularly the R-BCD algorithm, for finding their NEs. The $n$-sided PL condition is a reasonable extension of the gradient dominance condition, that holds in many existing problems and is better suited for a broader class of settings, particularly general-sum games. 
Furthermore, we studied the convergence rate of such first-order algorithm applied to general-sum games that satisfy the $n$-sided PL condition. Our analysis showed that the convergence rate depends on a local relation between the sum of the objective functions $F$, and the sum of their best responses $G_F$. 
Subsequently, we proposed two novel algorithms, IA-RBCD and A-RBCD, equipped with the gradient of $G_F$, that provably converge to a NE almost surely with random initialization, even if the individual objective functions are not lower bounded and have strict saddle points. We hope this work can shed some light on the understanding of a broader class of nonconvex optimization problems.


%% file: appendix.tex
\clearpage
\appendix
\begin{center}
\Large{Appendix}    
\end{center}

\section{Technical Lemmas}\label{sec:tech}
\begin{lemma}\cite{karimi2016linear}.\label{EB_PL}
    If $f(\cdot)$ is $l$-smooth and it satisfies PL with constant $\mu$, then it also satisfies error bound (EB) condition with $\mu$, i.e. 
    \begin{equation*}
        \|\nabla f(x)\|\geq\mu\|x_p-x\|, \forall x,
    \end{equation*}
    where $x_p$ is the projection of $x$ onto the optimal set, also it satisfies quadratic growth (QG) condition with $\mu$, i.e.
    \begin{equation*}
        f(x)-\min_y f(y)\geq \frac{\mu}{2}\|x_p-x\|^2, \forall x.
    \end{equation*}
    Conversely, if $f(\cdot)$ is $l$-smooth and satisfies EB with constant $\mu$, then it satisfies PL with constant $\frac{\mu}{l}$.
\end{lemma}

\begin{lemma}\label{L_mu}
    If $f(\cdot)$ is $L$-smooth and it satisfies $n$-sided $\mu$-PL condition, then $L\geq\mu$.
\end{lemma}
\begin{proof}
    From $L$-smoothness, we have
    \begin{align*}
        \|\nabla_i f(x_i,x_{-i})-\nabla_i f(y_i,x_{-i})\|\leq \|\nabla f(x_i,x_{-i})-\nabla f(y_i,x_{-i})\|\leq L\|x_i-y_i\|, \forall x_i,\ y_i.
    \end{align*}
    It indicates,
    \begin{align*}
        f(y_i,x_{-i})-f(x_i,x_{-i})\leq \langle\nabla_i f(x_i,x_{-i}),y_i-x_i\rangle+\frac{L}{2}||x_i-y_i||^2.
    \end{align*}
    Let $y_i=x_i-\nabla_i f(x_i,x_{-i})/L$. This leads to
    \begin{align*}
        f(x)-f(x^*_i(x),x_{-i})\geq \frac{1}{2L}||\nabla_i f(x)||^2.
    \end{align*}
    On the other hand, from the $n$-side PL, we get
    \begin{align*}
        f(x)-f(x^*_i(x),x_{-i})\leq \frac{1}{2\mu}||\nabla_i f(x)||^2.
    \end{align*}
    Putting the above inequalities together concludes the result. 
\end{proof}

\begin{lemma}\label{f_G_nabla}
    If $\{f_i(\cdot)\}$ are $L$-smooth and satisfy $n$-sided $\mu$-PL condition, then 
    \begin{align*}
        \frac{1}{2L}\sum_{i=1}^n\|\nabla_i f_i(x)\|^2\leq F(x)-G_F(x)\leq \frac{1}{2\mu}\sum_{i=1}^n\|\nabla_i f_i(x)\|^2.
    \end{align*}
    \begin{proof}
        This is a direct corollary from the last two inequalities of \cref{L_mu}.
    \end{proof}
\end{lemma}
This results implies that $\mu\leq L$. 

\section{Block Coordinate Descent (BCD) Algorithm}

\begin{algorithm}[H]
\caption{Cyclic Block Coordinate Descent (BCD)}\label{alg:BCD}
\begin{algorithmic}
   \STATE {\bfseries Input:} initial point $x^0=(x_1^0,...,x_n^0)$, learning rates $\alpha$
   \FOR{$t=1$ {\bfseries to} $T$}
   \FOR{$i=1$ {\bfseries to} $n$}
   \STATE $x_i^t=x_i^{t-1}-\alpha\nabla_{i}f_i(x_{1:i-1}^t,x_{i:n}^{t-1})$
   \ENDFOR
   \ENDFOR
\end{algorithmic}
\end{algorithm}

\section{Additional Examples}\label{app:example}

\subsection{Function $f^{(1)}(x_1,x_2)=(x_1-1)^2(x_2+1)^2+(x_1+1)^2(x_2-1)^2$}
Due to symmetry, we only show the condition for the first coordinate. 
\begin{align*}
    &\nabla_{x_1} f^{(1)}(x_1,x_2) = 2(x_1-1)(x_2+1)^2+2(x_1+1)(x_2-1)^2=4x(x_2^2+1)-8x_2,\\
    & f^*_{x_2}=2(x_2^2-1)^2/(x_2^2+1),\\
    & G_{f^{(1)}}(x_1,x_2)=\frac{(x_1^2-1)^2}{x_1^2+1}+\frac{(x_2^2-1)^2}{x_2^2+1},\\
    &\nabla G_{f^{(1)}}(x_1,x_2)=\Big(\dfrac{2x_1\left(x_1^2-1\right)\left(x_1^2+3\right)}{\left(x_1^2+1\right)^2},\dfrac{2x_2\left(x_2^2-1\right)\left(x_2^2+3\right)}{\left(x_2^2+1\right)^2}\Big)
\end{align*}
Thus, the $2$-sided PL holds iff $\exists \mu>0$, s.t. for all $x_1$ and $x_2$
\begin{align*}
  &2\Big((x_1-1)(x_2+1)^2+(x_1+1)(x_2-1)^2\Big)^2 -\mu \Big((x_1-1)^2(x_2+1)^2+(x_1+1)^2(x_2-1)^2-2\frac{(x_2^2-1)^2}{x_2^2+1}\Big)\geq0.
\end{align*}
The left-hand side is a quadratic equation with respect to $x_1$ and for $\mu=2$, it is
\begin{align*}
    &\Big((x_2+1)^2+(x_2-1)^2-1\Big)\Big(x_1^2\Big((x_2+1)^2+(x_2-1)^2\Big)-2x_1\Big((x_2+1)^2-(x_2-1)^2\Big)\Big)\\
    &+\Big((x_2+1)^2+(x_2-1)^2-1\Big)\Big((x_2+1)^2+(x_2-1)^2-4\frac{(x_2-1)^2(x_2+1)^2}{(x_2-1)^2+(x_2+1)^2}\Big).
\end{align*}
The above expression is positive for all $x$ and $y$.

\paragraph{Analysis of the origin:} Although, the origin point is a stationary point of $f_1$ since the Hessian at this point is not positive semi-definite, it is not a local minimum. However, it is straightforward to see that $(0,0)$ is in fact a NE of $f_1(x,y)$. Note that the Hessian at the origin is  
\begin{align*}
    H_f(0,0)=\left[\begin{matrix}
        4 & -8\\
        -8 & 4
    \end{matrix}\right]\not\succeq0.
\end{align*}

\subsection{Function $f^{(2)}(x_1,x_2)=(x_1-1)^2(x_2+1)^2+(x_1+1)^2(x_2-1)^2+\exp{-(x_2-1)^2}$}
For this function, we have
\begin{align*}
     \nabla_{x_1} f^{(2)}(x_1,x_2)&=2(x_1-1)(x_2+1)^2+2(x_1+1)(x_2-1)^2,\\
     \nabla_{x_2} f^{(2)}(x_1,x_2)&=2(x_2-1)(x_1+1)^2+2(x_2+1)(x_1-1)^2-2(x_2-1)\exp(-(x_2-1)^2).
\end{align*}
and 
\begin{align*}
    \nabla_{x_1}^2 f^{(2)}(x_1,x_2)&=2(x_2+1)^2+2(x_2-1)^2\geq4,\\  
    \nabla_{x_2}^2 f^{(2)}(x_1,x_2)&=2(x_1+1)^2+2(x_1-1)^2+4(x_2-1)^2\exp(-(x_2-1)^2)-2\exp(-(x_2-1)^2)\geq2.
\end{align*}
It is straightforward to see that this function is smooth as the second-order derivatives are upper-bounded. Moreover, since both the second-order derivatives are strictly positive, then it is 2-sided PL. 
It is noteworthy that $(0,0)$ is also an NE for this function but it is not a local minimum as the Hessian at the origin is not positive semi-definite.

\begin{remark}
    The $n$-sided PL condition is defined coordinated-wise, with the coordinates aligned with the vectors $\{e_1,...,e_n\}$, where $e_i$ are standard basis vectors in $\mathbb R^d$. 
    This definition can naturally be extended to $n$-sided directional PL in which the $i$-th inequality is aligned with a designated vector $v_i$. In this extension, the partial gradient and $f^*_{x_{-i}}$ are replaced with their directional variants along vector $v_i$. 
    Note that the results of this work will remain valid in the directional setting, provided that the definitions of NE and the presented algorithms are adjusted to their respective directional variants. 
\end{remark}

\section{Technical Proofs}\label{sec:proofs}

\subsection{Proof of Lemma \ref{SP_equal_NE}}\label{example_function_2_sided}

$\nabla_{i}f_i(x_i,x_{-i})=0,\forall i\in[n]$ $\Longrightarrow$ Nash Equilibrium: 
If a point $x$ satisfies $\nabla_i f_i(x)=0$, $\forall i\in[n]$ from the definition of $n$-sided PL and $f_{i,x_{-i}}^\star$, we have
\begin{equation*}
\begin{aligned}
    &0=\|\nabla_{i} f_i(x)\|^2\geq 2\mu(f_i(x_i,x_{-i})-f_{i,x_{-i}}^\star)\geq 0,\quad \forall i\in[n],\\
    &\Longrightarrow f_i(x_i,x_{-i})=f_{i,x_{-i}}^\star=\min_{y_i}f_i(y_i,x_{-i}),\quad \forall i\in[n],\\
    &\Longrightarrow f_i(x_i,x_{-i})\leq f_i(\tilde{x_i},x_{-i}),\quad\forall \tilde{x_i},\forall i\in[n],
\end{aligned}
\end{equation*}
which means $x$ satisfies the definition of Nash Equilibrium.

If $f$ is differentiable, then Nash Equilibrium $\Longrightarrow$  $\nabla_{i}f_i(x_i,x_{-i})=0,\forall i\in[n]$ :
If a point $x$ is a Nash Equilibrium, then $f_i(x_i,x_{-i})\leq f_i(\tilde{x_i},x_{-i}),\forall \tilde{x_i},\forall i\in[n]$. Based on the first-order optimality condition, we have
\begin{equation*}
    \nabla_{i}f_i(x_i,x_{-i})=0,\quad \forall i\in[n].
\end{equation*}

\subsection{Proof of \cref{G_smooth}}\label{proof:G_smooth}

\begin{proof}
Based on the Lipschitzness of the $\nabla f_i$, we have that
\begin{equation*}
    \|\nabla_{i}f_i(x_i^\star(y),x_{-i})\|=\|\nabla_{i}f_i(x_i^\star(y),x_{-i})-\nabla_{i}f_i(x_i^\star(y),y_{-i})\|\leq L\|x_{-i}-y_{-i}\|.
\end{equation*}
Also, from $n$-sided PL condition and \cref{EB_PL},
\begin{equation*}
    \|\nabla_{i}f_i(x_i^\star(y),x_{-i})\| \geq \mu\|x_i^\star(y)-x_i^\star(x_i^\star(y),x_{-i})\|.
\end{equation*}
From the above inequalities, we imply
\begin{equation*}
    \|x_i^\star(y)-x_i^\star(x_i^\star(y),x_{-i})\|\leq\frac{L}{\mu}\|x_{-i}-y_{-i}\|.
\end{equation*}
Then, we can show the smoothness of $g_i(x):=f_i(x_i^\star(x),x_{-i})$. From the Lemma A.3. of \cite{sanjabi2018solving}, we get $\nabla f_i(x_i^\star(x),x_{-i})=\nabla f_i(x_i^\star(x_i^\star(y),x_{-i}),x_{-i})$. In consequence,
\begin{equation*}
\begin{aligned}
      &\|\nabla g_i(x)-\nabla g_i(y)\|
      =\|\nabla f_i(x_i^\star(x_i^\star(y),x_{-i}),x_{-i})-\nabla f_i(x_i^\star(y),y_{-i})\|\\
      &\leq \|\nabla f_i(x_i^\star(x_i^\star(y),x_{-i}),x_{-i})-\nabla f_i(x_i^\star(x_i^\star(y),x_{-i}),y_{-i})\|\\
      &+\|\nabla f_i(x_i^\star(x_i^\star(y),x_{-i}),y_{-i})-\nabla f_i(x_i^\star(y),y_{-i})\|\\
      &\leq L\|x_{-i}-y_{-i}\|+L\|x_i^\star(y)-x_i^\star(x_i^\star(y),x_{-i})\|\leq \left(L+\frac{L^2}{\mu}\right)\|x_{-i}-y_{-i}\|.
\end{aligned}
\end{equation*}
The first equality is due to Lemma A.5 in \cite{nouiehed2019solving}. 
This leads to
\begin{equation*}
\begin{aligned}
    \|\nabla G_F(x)-\nabla G_F(y)\|&=\|\nabla \sum_{i=1}^n g_i(x)-\nabla \sum_{i=1}^n g_i(y)\|\leq\sum_{i=1}^n\|\nabla g_i(x)-\nabla g_i(y)\|\\
    &\leq\sum_{i=1}^n \left(L+\frac{L^2}{\mu}\right)\|x_{-i}-y_{-i}\|\leq  n\left(L+\frac{L^2}{\mu}\right)\|x-y\|.
\end{aligned}
\end{equation*}

\end{proof}


\subsection{Proof of \cref{F-G_Nash}}
If $x^\star$ is an NE, then $f_i(x^\star)=\min_{x_i}f_i(x_i,x_{-i}^\star)$. By summing up over $i$, we have $F(x^\star)-G_F(x^\star)=0$.
If $F(x^\star)-G_F(x^\star)=0$, then $\sum_{i=1}^n [f_i(x^\star)-\min_{x_i}f_i(x_i,x_{-i}^\star)]=0$, which indicates $f_i(x^\star)=\min_{x_i}f_i(x_i,x_{-i}^\star)$ for all $i\in[n]$. Thus, $x\star$ is an NE.

From \cref{f_G_nabla}, we have $\|\nabla _i f_i(x)\|\leq\sqrt{\sum_{i=1}^n\|\nabla_i f_i(x)\|^2}\leq\sqrt{2L(F(x)-G_F(x))}\leq \sqrt{2L\epsilon}$, $\forall i\in[n]$.


\subsection{Proof of \cref{convergence_ideal_kt0}}\label{proof:convergence_ideal_kt0}

Note that function $F(x)-G_F(x)$ is $n(L+L')$-smooth. Furthermore,  $n(L+L')$-smoothness implies the $n(L+L')$-coordinate-wise smoothness. Thus, for $\alpha\leq\frac{1}{L}$, we get
\begin{equation*}
\begin{aligned}
         F(x^{t+1})-G_F(x^{t+1})\leq&F(x^{t})-G_F(x^{t})+\langle\nabla_{i^t} (F(x^t)-G_F(x^{t})),x_{i^t}^{t+1}-x_{i^t}^{t}\rangle+\frac{n(L+L')}{2}\|x_{i^t}^{t+1}-x_{i^t}^{t}\|^2\\
         &\leq F(x^{t})-G_F(x^{t})-\alpha\langle\nabla_{i^t} (F(x^t)-G_F(x^{t})),\nabla_{i^t} f_{i^t}(x^t)\rangle+\frac{n(L+L')\alpha^2}{2}\|\nabla_{i^t} f_{i^t}(x^t)\|^2\\
         &= F(x^{t})-G_F(x^{t})-\alpha\langle\nabla_{i^t} (F_{-i^t}(x^t)-G_F(x^{t})),\nabla_{i^t} f_{i^t}(x^t)\rangle-\left(\alpha-\frac{n(L+L')\alpha^2}{2}\right)\|\nabla_{i^t} f_{i^t}(x^t)\|^2\\
\end{aligned}
\end{equation*}
By taking the expectation over $i^t$, 
\begin{equation*}
\begin{aligned}
        \mathbb{E}[ F(x^{t+1})-G_F(x^{t+1})|x^t]\leq&F(x^{t})-G_F(x^{t})-\frac{1}{n}\alpha\sum_{i=1}^n\langle\nabla_{i} (F_{-i}(x^t)-G_F(x^{t})),\nabla_{i} f_{i}(x^t)\rangle\\
        &-\frac{1}{n}\left(\alpha-\frac{n(L+L')\alpha^2}{2}\right)\sum_{i=1}^n\|\nabla_{i} f_{i}(x^t)\|^2,\\
        &\leq F(x^{t})-G_F(x^{t})+\frac{1}{n}\alpha\kappa\sum_{i=1}^n \|\nabla_i f_i(x^t)\|^2-\frac{1}{n}\left(\alpha-\frac{n(L+L')\alpha^2}{2}\right)\sum_{i=1}^n\|\nabla_{i} f_{i}(x^t)\|^2,\\
        &=F(x^{t})-G_F(x^{t})-\frac{1}{n}\left((1-\kappa)\alpha-\frac{n(L+L')\alpha^2}{2}\right)\sum_{i=1}^n\|\nabla_{i} f_{i}(x^t)\|^2,\\
        &\leq \left(1-\frac{(1-\kappa)\mu\alpha}{2n}\right)(F(x^{t})-G_F(x^{t})),
\end{aligned}
\end{equation*}

where the second line comes from the condition $\sum_{i=1}^n\langle\nabla_{i} (G_F(x^{t})-F_{-i}(x^t)),\nabla_{i} f_{i}(x^t)\rangle\leq \kappa\sum_{i=1}^n \|\nabla_i f_i(x^t)\|^2$ of case 1 and the last line is due to $\alpha\leq\frac{1-\kappa}{n(L+L')}$.

\hfill\qedsymbol{}



\subsection{Proof of \cref{convergence_ideal}}\label{proof:convergence_ideal}

\textbf{Case 1:} This is analogous to the proof of Theorem \ref{convergence_ideal_kt0}.

\noindent
\textbf{Case 2:} Recall that $F(x)-G_F(x)$ is $n(L+L')$-smooth which leads to
\begin{equation*}
\begin{aligned}
    F(x^{t+1})-G_F(x^{t+1})\leq&F(x^{t})-G_F(x^{t})+\langle\nabla_{i^t} (F(x^t)-G_F(x^{t})),x_{i^t}^{t+1}-x_{i^t}^{t}\rangle+\frac{n(L+L')}{2}\|x_{i^t}^{t+1}-x_{i^t}^{t}\|^2\\
    =& F(x^{t})-G_F(x^{t})-\alpha\langle\nabla_{i^t} (F(x^{t})-G_F(x^{t})),\nabla_{i^t} f_{i^t}(x^t)+k^t(\nabla_{i^t} (G_F(x^t)-F_{-i^t}(x^t)))\rangle\\
    &+\frac{n(L+L')\alpha^2}{2}\|\nabla_{i^t} f_{i^t}(x^t)+k^t(\nabla_{i^t} (G_F(x^t)-F_{-i^t}(x^t)))\|^2\\
    =& F(x^{t})-G_F(x^{t})-(\alpha-\frac{n(L+L')\alpha^2}{2})\|\nabla_{i^t}f_{i^t}(x^t)\|^2\\
    &+(\alpha-\alpha k^t+n(L+L')\alpha^2k^t)\langle\nabla_{i^t} f_{i^t}(x^t),\nabla_{i^t} (G_F(x^t)-F_{-i^t}(x^t))\rangle\\
    &+\left(\alpha k^t+\frac{n(L+L')\alpha^2(k^t)^2}{2}\right)\|\nabla_{i^t} (G_F(x^t)-F_{-i^t}(x^t))\|^2.
\end{aligned}
\end{equation*}
Taking the expectation over $i^t$, we have
\begin{equation*}
\begin{aligned}
    \mathbb{E}[F(x^{t+1})-G_F(x^{t+1})|x^t]\leq &F(x^t)-G_F(x^t)-\frac{1}{n}(\alpha-\frac{n(L+L')\alpha^2}{2})\sum_{i=1}^n\|\nabla_{i}f_{i}(x^t)\|^2\\
    &+\frac{1}{n}(\alpha-\alpha k^t+n(L+L')\alpha^2k^t)\sum_{i=1}^n\langle\nabla_{i} f_{i}(x^t),\nabla_{i} (G_F(x^t)-F_{-i}(x^t))\rangle\\
    &+\frac{1}{n}\left(ak^t+\frac{n(L+L')\alpha^2(k^t)^2}{2}\right)\sum_{i=1}^n\|\nabla_{i} (G_F(x^t)-F_{-i}(x^t))\|^2.
\end{aligned}
\end{equation*}
Let
\begin{equation*}
\begin{aligned}
    h(k^t):=&\frac{1}{n}(\alpha-\alpha k^t+n(L+L')\alpha^2k^t)\sum_{i=1}^n\langle\nabla_{i} f_{i}(x^t),\nabla_{i} (G_F(x^t)-F_{-i}(x^t))\rangle\\
    &+\frac{1}{n}\left(ak^t+\frac{n(L+L')\alpha^2(k^t)^2}{2}\right)\sum_{i=1}^n\|\nabla_{i} (G_F(x^t)-F_{-i}(x^t))\|^2,
\end{aligned}
\end{equation*}
Note that this is a convex function of $k^t$ and for $k^t=-1$, we have
\begin{align}\label{h_1}
    h(-1)&=-\frac{2\alpha-n(L+L')\alpha^2}{2n}\|\nabla F(x^t)-\nabla G_F(x^t)\|^2+\frac{1}{n}\left(\alpha-\frac{n(L+L')\alpha^2}{2}\right)\sum_{i=1}^n\|\nabla_i f_i(x^t)\|^2\\
    &\leq \frac{1}{n}\left(\alpha-\frac{n(L+L')\alpha^2}{2}\right)\sum_{i=1}^n\|\nabla_i f_i(x^t)\|^2.
\end{align}
The above inequality holds when $\frac{2\alpha-n(L+L')\alpha^2}{2n}\geq0$.
Let 
$$
A:=\sum_{i=1}^n\langle\nabla_{i} f_{i}(x^t),\nabla_{i} (G_F(x^t)-F_{-i}(x^t))\rangle, \quad 
B:=\sum_{i=1}^n\|\nabla_{i} (G_F(x^t)-F_{-i}(x^t))\|^2.
$$
Also, the function value $h(k^t)$ at minimizer $k^\star=-\frac{(n(L+L')\alpha-1)A+B}{n(L+L')\alpha B}$ is less or equals to zero if

\begin{equation}\label{alphat_f_G1}
    n^2(L+L')^2(A)^2\alpha^2-2n(L+L')(A)^2\alpha+(B-A)^2\geq 0,\Longrightarrow\alpha\leq\frac{1}{2n(L+L')}\frac{(B-A)^2}{(A)^2},
\end{equation}
In the case 2, since $\frac{(B-A)^2}{(A)^2}\geq C$, \cref{alphat_f_G1} is satisfied if
\begin{equation*}
    \alpha\leq\frac{C}{2n(L+L')}.
\end{equation*}

From the convexity of $h$, if $\alpha\leq\frac{C}{2n(L+L')}$, i.e., $h(k^*)\leq0$, then combining \cref{h_1} and \cref{alphat_f_G1},
\begin{equation}\label{h_lambda}
    h(-\lambda+(1-\lambda)k^\star)\leq\lambda h(-1)+(1-\lambda)h(k^\star)\leq\frac{\lambda}{n}\left(\alpha-\frac{nL\alpha^2}{2}-\frac{nL'\alpha^2}{2}\right)\sum_{i=1}^n\|\nabla_i f_i(x^t)\|^2,
\end{equation}
for all $\lambda\in[0,1]$. 

By setting $k^t=-2+\frac{A}{B}=-\lambda+(1-\lambda)k^\star$, we have
\begin{equation*}
    0\leq \lambda=1-\frac{n(L+L')\alpha(k^t+1)B}{(1-n(L+L')\alpha)(A-B)}= 1- \frac{n(L+L')\alpha}{1-n(L+L')\alpha}<1,
\end{equation*}
Thus, from \cref{h_lambda},
\begin{equation}\label{h_kt}
    h(k^t)= h(-\lambda+(1-\lambda)k^\star)\leq\frac{1}{n}\Big(1- \frac{n(L+L')\alpha}{1-n(L+L')\alpha}\Big)\left(\alpha-\frac{n(L+L')\alpha^2}{2}\right)\sum_{i=1}^n\|\nabla_i f_i(x^t)\|^2.
\end{equation}
As a result,
\begin{equation*}
\begin{aligned}
    &\mathbb{E}[F(x^{t+1})-G_F(x^{t+1})|x^t]\\
    &\leq Fx^t)-G_F(x^t)-\frac{1}{n}\left(\alpha-\frac{n(L+L')\alpha^2}{2}\right)\sum_{i=1}^n\|\nabla_i f_i(x^t)\|^2+h(k^t)\\
    &\leq F(x^t)-G_F(x^t)-\frac{1}{n}\frac{n(L+L')\alpha}{1-n(L+L')\alpha}\left(\alpha-\frac{n(L+L')\alpha^2}{2}\right)\sum_{i=1}^n\|\nabla_i f_i(x^t)\|^2\\
    &\leq F(x^t)-G_F(x^t)-\frac{1}{2n}\frac{n(L+L')\alpha^2}{1-n(L+L')\alpha}\sum_{i=1}^n\|\nabla_i f_i(x^t)\|^2\\
    &\leq \left(1-\frac{(L+L')\mu\alpha^2}{1-n(L+L')\alpha}\right)(F(x^t)-G_F(x^t))\leq \left(1-\frac{(L+L')\mu\alpha^2}{2}\right)(F(x^t)-G_F(x^t)).
\end{aligned}
\end{equation*}

where we apply \cref{h_kt} in the third line, \cref{f_G_nabla} in the fifth line. 

\textbf{Case 3:} In this case, we have $x_{i^t}^{t+1}-x_{i^t}^t=-\alpha (\nabla_{i^t}F(x^t)-\nabla_{i^t}G_F(x^t))$,
\begin{equation*}
\begin{aligned}
    &\mathbb{E}[F(x^{t+1})-G_F(x^{t+1})|x^t]\\
    &\leq F(x^{t})-G_F(x^{t})+\mathbb{E}[\langle\nabla_{i^t}F(x^t)-\nabla_{i^t}G_F(x^t),x_{i^t}^{t+1}-x_{i^t}^t\rangle+\frac{n(L+L')}{2}\|x_{i^t}^{t+1}-x_{i^t}^t\|^2|x^t]\\
    &=  F(x^{t})-G_F(x^{t})-\left(\alpha-\frac{n(L+L')\alpha^2}{2}\right)\mathbb{E}[\|\nabla_{i^t}F(x^t)-\nabla_{i^t}G_F(x^t)\|^2|x^t]\\
    &\leq F(x^{t})-G_F(x^{t})-\frac{1}{2}\alpha\mathbb{E}[\|\nabla_{i^t}F(x^t)-\nabla_{i^t}G_F(x^t)\|^2|x^t]\\
    &\leq F(x^{t})-G_F(x^{t})-\frac{1}{2n}\alpha\|\nabla F(x^t)-\nabla G_F(x^t)\|^2\\
    &\leq F(x^{t})-G_F(x^{t})-\frac{\alpha\nu}{n}(F(x^{t})-G_F(x^{t}))^{\frac{2}{\theta}}.
\end{aligned}
\end{equation*}

From the Lemma 6 of \cite{fatkhullin2022sharp}, we have

\begin{equation*}
 \mathbb{E}[F(x^{t+k})-G_F(x^{t+k})|x^t]\leq\frac{(2n)^{\frac{\theta}{2-\theta}}{\frac{2-\theta}{\theta}}^{-\frac{\theta+2}{2-\theta}}+n^\frac{\theta}{2-\theta}\theta^{-\frac{\theta}{2-\theta}}+({\nu\alpha})^{\frac{\theta}{2-\theta}}(F(x^{t})-G_F(x^{t}))}{({\nu\alpha}(k+1))^{\frac{\theta}{2-\theta}}}   
\end{equation*}

\hfill\qedsymbol{}

\subsection{Proof of \cref{case_many_enough}}

As $B=\sup_{t\to+\infty}\frac{b(t)}{t}<1$, for all $0<\epsilon<1-B$, there exists $\hat{t}$ and such that $\frac{b(t)}{t}<B+\epsilon<1$, $\forall t\geq \hat{t}$. Then, by denoting $T_1(t)$ and $T_2(t)$ as the number of times entering Case 1 and Case 2 until time $t$, we have

\begin{equation*}
\begin{aligned}
\mathbb{E}[F(x^t)-G_F(x^t)]&\leq (1-\frac{(1-\kappa)\mu\alpha}{2})^{T_1(t)}(1-\frac{(L+L')\mu\alpha^2}{2})^{T_2(t)}(F(x^0)-G_F(x^0))\\
&\leq\Big(\min\big\{1-\frac{(1-\kappa)\mu\alpha}{2},1-\frac{(L+L')\mu\alpha^2}{2}\big\}\Big)^{T_1(t)+T_2(t)}(F(x^0)-G_F(x^0))\\
&=\Big(\min\big\{1-\frac{(1-\kappa)\mu\alpha}{2},1-\frac{(L+L')\mu\alpha^2}{2}\big\}\Big)^{t-b(t)}(F(x^0)-G_F(x^0))\\
&\leq\Big(\min\big\{1-\frac{(1-\kappa)\mu\alpha}{2},1-\frac{(L+L')\mu\alpha^2}{2}\big\}\Big)^{(1-B-\epsilon)t}(F(x^0)-G_F(x^0)).
\end{aligned}
\end{equation*}


\subsection{Proof of Theorem \ref{thm:measure}}

\begin{proof}
    Suppose both cases 1 and 2 do not hold, then the iterates satisfy,
    \begin{equation*}
    \frac{(B-A)^2}{(A)^2}< C,\ \mathrm{and}\ A> \gamma D.   
    \end{equation*}
    where $ A=\sum_{i=1}^n\langle\nabla_{i} f_{i}(x^t),\nabla_{i}G_F(x^t)-\nabla_iF_{-i}(x^t)\rangle$, $B=\sum_{i=1}^n\|\nabla_{i} G_F(x^t)-\nabla_i F_{-i}(x^t)\|^2$, and $D=\sum_{i=1}^n \|\nabla_i f_i(x^t)\|^2$.
    
    By simplifying the first equation and consider $A>\gamma D\geq 0$, we have
    \begin{equation*}
    \begin{aligned}
    &A>\gamma\sum_{i=1}^n \|\nabla_i f_i(x^t)\|^2,
     &\frac{1}{1+\sqrt{C}}B<A<\frac{1}{1-\sqrt{C}}B.    
    \end{aligned}
    \end{equation*}
    In consequence,
    \begin{equation*}
    \begin{aligned}
        \|\nabla F(x^t)-\nabla G_F(x^t)\|^2&=\sum_{i=1}^n\|\nabla_i F(x^t)-\nabla_i G_F(x^t)\|^2\\
        &=\sum_{i=1}^n\|\nabla_i f_i(x^t)+\nabla_i F_{-i}(x^t)-\nabla_i G_F(x^t)\|^2= D+ B-2A \\
        &\leq D+ \frac{2\sqrt{C}-1}{1-\sqrt{C}}A= (1+\frac{2\sqrt{C}-1}{1-\sqrt{C}}\gamma)D\leq (2+\frac{4\sqrt{C}-2}{1-\sqrt{C}}\gamma)L(F(x^t)-G_F(x^t)).
    \end{aligned}
    \end{equation*}
    The last inequality uses the smoothness of $F$. Moreover, $F(x^t)-G_F(x^t)$ satisfies,
    \begin{equation*}
    \begin{aligned}
        F(x^t)-G_F(x^t)&< \frac{\|\nabla F(x^t)-\nabla G_F(x^t)\|^\theta}{(2\nu)^{\theta/2}}< (L+\frac{2\sqrt{C}-1}{1-\sqrt{C}}\gamma L)^{\theta/2}(F(x^t)-G_F(x^t))^{\theta/2}.
    \end{aligned}
    \end{equation*}
    The above inequality brings the upper bound for $F(x^t)-G_F(x^t)$ and $\sum_{i=1}^n\|\nabla_i f_i(x^t)\|^2$,
    \begin{equation*}
        F(x^t)-G_F(x^t)< (L+\frac{2\sqrt{C}-1}{1-\sqrt{C}}\gamma L)^{\frac{\theta}{2-\theta}},
    \end{equation*}
    and
    \begin{equation*}
        \sum_{i=1}^n\|\nabla_i f_i(x^t)\|^2\leq 2L(F(x^t)-G_F(x^t))< 2L^{\frac{2}{2-\theta}}(1+\frac{2\sqrt{C}-1}{1-\sqrt{C}}\gamma )^{\frac{\theta}{2-\theta}}.
    \end{equation*}
    As $C\to 0$ and $\gamma\to 1$, $F(x^t)-G_F(x^t)<\epsilon$ , $\forall\epsilon>0$. Thus, for all the non-NE points (i.e. $\{x|F(x)-G_F(x)>0\}$),  the measure of the point that satisfies case 3 converges to $0$.
\end{proof}


\subsection{Upper bound of $\sum_{i=1}^n\|\nabla_{i} G_F(x)-\nabla_i F_{-i}(x)\|^2$}\label{proof:G_deri_bound}
\begin{lemma}\label{G_deri_bound}
   For $n$-sided $\mu$-PL functions $\{f_i(x)\}$ satisfying \cref{L_Lip}, then $\sum_{i=1}^n\|\nabla_{i} G_F(x)-\nabla_i F_{-i}(x)\|^2\leq \frac{3nL^2}{\mu^2}\sum_{i=1}^n\|\nabla_i f_i(x)\|^2$, for all $x$.
\end{lemma}

\begin{proof}
\begin{equation}
\begin{aligned}
      &\sum_{i=1}^n\|\nabla_{i} G_F(x)-\nabla_i F_{-i}(x)\|^2=\sum_{i=1}^n \|\nabla_{i} G_F(x)-\nabla_i F(x^t)+\nabla_i f_i(x)\|^2\\
   &\leq 2\sum_{i=1}^n  \|\nabla_{i} G_F(x)-\nabla_i F(x)\|^2+2\sum_{i=1}^n \|\nabla_i f_i(x)\|^2\\
   &= 2\|\nabla G_F(x)-\nabla F(x)\|^2+2\sum_{i=1}^n \|\nabla_i f_i(x)\|^2\\
   &= 2\|\sum_{i=1}^n\nabla f_i(x_i^\star(x),x_{-i})-\sum_{i=1}^n\nabla f_i(x_i)\|^2+2\sum_{i=1}^n\|\nabla_i f_i(x)\|^2\\
   &\leq 2n\sum_{i=1}^n\|\nabla f_i(x_i^\star(x),x_{-i})-\nabla f_i(x_i)\|^2+2\sum_{i=1}^n\|\nabla_i f_i(x)\|^2\\
   &\leq 2nL^2 \sum_{i=1}^n\|x_i^\star(x)-x_i\|^2+2\sum_{i=1}^n\|\nabla_i f_i(x)\|^2\\
   &\leq (\frac{2nL^2}{\mu^2}+2)\sum_{i=1}^n\|\nabla_i f_i(x)\|^2\leq \frac{3nL^2}{\mu^2}\sum_{i=1}^n\|\nabla_i f_i(x)\|^2
\end{aligned}
\end{equation}
The last inequality is due to Lemma \ref{f_G_nabla} that $l\geq\mu$ and  $n\geq 2$.
\end{proof}

\subsection{Proof of Theorem \ref{convergence_practical}}\label{app:p_convergence_practical}

\begin{theorem}
For the $n$-sided $\mu$-PL functions $\{f_i(x)\}$ satisfying \cref{L_Lip}, by implementing \cref{arbcd_practical} with $\beta\leq\frac{1}{L}$ and $T'\geq\log\big(\frac{nL^2}{\mu^2\alpha^4}\big)/\log(\frac{1}{1-\mu\beta})$, 
\begin{itemize}[leftmargin=7pt]
\item in Case 1 with $\alpha\!\leq\!\frac{1-\kappa}{n(L+L')}$, we have
$
\mathbb{E}[F(x^{t+1})-G_F(x^{t+1})|x^t]\leq (1-\frac{\mu\alpha(1-\gamma)}{2})(f(x^t)-G_F(x^t)),
$

\item in Case 2 with $\alpha\leq\min\big\{(\frac{\gamma}{6n})^{1/2},(\frac{\gamma^2C}{4(5n+4)}),\frac{7(L+L')\mu^2}{24L^2},(\frac{\mu^3}{2450nL^3})^{1/2},(\frac{\mu^3}{480nL^3})^{1/4}\big\}$, we have
\begin{equation*}
    \mathbb{E}[F(x^{t+1})-G_F(x^{t+1})|x^t]\leq C_0\cdot(f(x^t)-G_F(x^t)),
\end{equation*}

where $C_0:=(1-\frac{(L+L')\mu\alpha^2}{4})$.

\item in Case 3 with  $f\!-\!G_F$ is non-increasing when $\alpha\!\leq\!\min\!\big\{\frac{\|\nabla F(x^t)-\nabla G_F(x^t)\|}{12\sqrt{\sum_{i=1}^n\|\nabla_i f_i(x^t)\|^2}},\frac{1}{n(L+L')}\big\}$. Furthermore, if $f\!-\!G_F$ satisfies $(\theta,\nu)$-PL condition and case 3 occurs from iterates $t$ to $t+k$, then
\begin{equation*}
    \mathbb{E}[F(x^{t+k})-G_F(x^{t+k})|x^t]\leq\mathcal{O}\Big(\frac{f(x^{t})-G_F(x^{t})}{k^{\frac{\theta}{2-\theta}}}\Big)
\end{equation*}
\end{itemize}
\end{theorem}

\begin{proof}
To approximate $G_F(x^t)$, we need to estimate the best response of i-th block $x_i^\star(x^t)$ when other blocks are fixed. As the function $\{f_i\}$ satisfies $n$-sided PL condition, the function $f_i(x_i)=f_i(x_i;x_{-i}^t)$ satisfies strong PL condition. Therefore, the best response can be approximated efficiently by applying the gradient descent with partial gradient $\nabla_i f_i(x_i,x_{-i}^t)$. For any $\delta>0$,
\begin{equation}\label{best_response_distance1}
\begin{aligned}
    \|x_i^{\star}(x^t)-y_{i}^{t,T'}\|^2&\leq \frac{2}{\mu}(f_i(y_{i}^{t,T'},x_{-i}^t)-\min_{x_i}f_i(x_i,x_{-i}^t))\\
    &\leq \frac{2}{\mu}(1-\mu\beta)^{T'}(f_i(x^t)-\min_{x_i}f_i(x_i,x_{-i}^t))\\
    &\leq \frac{1}{\mu^2}(1-\mu\beta)^{T'} \|\nabla_i f_i(x^t)\|^2\leq \frac{\delta^2}{nL^2} \|\nabla_i f_i(x^t)\|^2.
\end{aligned}
\end{equation}
if $T'\geq  \frac{1}{\log(\frac{1}{1-\mu\beta})}\log(\frac{nL^2}{\mu^2\delta^2})$ and $\beta\leq\frac{1}{L}$.  The first inequality comes from the quadratic growth properties of the function $f_i(x_i)=f_i(x_i,x_{-i}^t)$ since it satisfies the strong PL condition. The second inequality comes from the convergence of gradient descent under the PL condition. The third inequality comes from the definition of the n-sided PL condition.
\begin{equation}\label{G_tildeG_diff1}
\begin{aligned}
    \|\nabla G_F(x^t)- \tilde\nabla G_F(x^t)\|&=\left\|\sum_{i=1}^n\nabla f_i(x_i^\star(x^t),x_{-i})-\sum_{i=1}^n\nabla f_i(y_i^{t,T'},x_{-i})\right\|\\
    &\leq\sum_{i=1}^n\left\|\nabla f_i(x_i^\star(x^t),x_{-i})-\nabla f_i(y_i^{t,T'},x_{-i})\right\|\leq L\sum_{i=1}^n \left\|x_i^{\star}(x^t)-y_{i}^{t,T'}\right\|\\
    &\leq \frac{\delta}{\sqrt{n}}\sum_{i=1}^n\|\nabla_i f_i(x^t)\|\leq \delta \sqrt{\sum_{i=1}^n\|\nabla_i f_i(x^t)\|^2}.
\end{aligned}
\end{equation}
where we apply the \cref{best_response_distance1}.

The second line comes from the triangle inequality, and the third line is due to the $L$-smoothness of $\nabla f(x^t)$. Then, we denotes $\bar x^{t+1}$ as the iterates in the ideal case, i.e.
\begin{equation}\label{bar_x1}
   \bar x_{i}^{t+1}=\left\{\begin{aligned}
    &x_{i}^t-\alpha(\nabla_{i} f_{i}(x^t)+k^t(\nabla_{i} (G_F(x^t)-F_{-i}(x^t)))),& \mathrm{if} i=i^t,\\
    &x_{i}^{t+1},& \mathrm{if} i\neq i^t.
\end{aligned} \right. 
\end{equation}
Next, by choosing $\delta=\alpha^2$ we show the convergence of $f(x^t)-G_F(x^t)$. To do so, we break it into different cases.

\textbf{Case 1:}
If $\sum_{i=1}^n\langle\tilde\nabla_i G_F(x^t)-\nabla_i F_{-i}(x^t),\nabla_i f_i(x^t)\rangle\leq (\gamma-\delta)\sum_{i=1}^n\|\nabla_i f_i(x^t)\|^2$, we have
\begin{equation*}
\begin{aligned}
&\sum_{i=1}^n\langle\tilde\nabla_i G_F(x^t)-\nabla_i F_{-i}(x^t),\nabla_i f_i(x^t)\rangle\\
&=\sum_{i=1}^n\langle\nabla_i G_F(x^t)-\tilde\nabla_i G_F(x^t),\nabla_i f_i(x^t)\rangle+\sum_{i=1}^n\langle\tilde\nabla_i G_F(x^t)-\nabla_i F_{-i}(x^t),\nabla_i f_i(x^t)\rangle\\
&\leq\sum_{i=1}^n\|\nabla_i G_F(x^t)-\tilde\nabla_i G_F(x^t)\|\|\nabla_i f_i(x^t)\|+\sum_{i=1}^n\langle\tilde\nabla G_F(x^t)-\nabla_i F_{-i}(x^t),\nabla_i f_i(x^t)\rangle\\
&\leq\|\nabla G_F(x^t)-\tilde\nabla G_F(x^t)\|\sqrt{\sum_{i=1}^n\|\nabla_i f_i(x^t)\|^2}+\sum_{i=1}^n\langle\tilde\nabla G_F(x^t)-\nabla_i F_{-i}(x^t),\nabla_i f_i(x^t)\rangle\\
&\leq \delta\sum_{i=1}^n\|\nabla_i f_i(x^t)\|^2+\sum_{i=1}^n\langle\tilde\nabla G_F(x^t)-\nabla_i F_{-i}(x^t),\nabla_i f_i(x^t)\rangle\leq \gamma\sum_{i=1}^n\|\nabla_i f_i(x^t)\|^2.
\end{aligned}
\end{equation*}

where we apply the Cauchy-Schwarz inequality at the third and fourth lines. 

By choosing $k^t=0$, from \cref{convergence_ideal_kt0}, we have
\begin{equation*}
\begin{aligned}
&\mathbb{E}[f(x^{t+1})-G_F(x^{t+1})|x^t]=\mathbb{E}[f(\bar x^{t+1})-G_F(\bar x^{t+1})|x^t]\leq  \left(1-\frac{(1-\kappa)\mu\alpha}{2n}\right)(f(x^t)-G_F(x^t)).
\end{aligned}
\end{equation*}

\textbf{Case 2:} $\left(\frac{\tilde B}{\tilde A}-1\right)^2\geq C$ and $\tilde A\geq (\gamma-\delta)\sum_{i=1}^n\|\nabla_i f_i(x^t)\|^2$, where $\tilde A=\sum_{i=1}^n\langle\nabla_{i} f_{i}(x^t),\nabla_{i}\tilde G_F(x^t)-\nabla_iF_{-i}(x^t)\rangle$ and $\tilde B=\sum_{i=1}^n\|\nabla_{i}\tilde G_F(x^t)-\nabla_i F_{-i}(x^t)\|^2$.
From the assumption of case 2, we have
\begin{align}\label{b_t_f}
        \tilde A\geq (\gamma-\delta)\sum_{i=1}^n\|\nabla_i f_i(x^t)\|^2\    \implies  \sqrt{\tilde B} \geq \Big(\gamma-\delta\Big)\sqrt{\sum_{i=1}^n\|\nabla_i f_i(x^t)\|^2}.
\end{align}
Then, we need to bound the difference between $\sum_{i=1}^n\|\nabla_{i}\tilde G_F(x^t)-\nabla_i F_{-i}(x^t)\|$ and $\sum_{i=1}^n\|\nabla_{i} G_F(x^t)-\nabla_i F_{-i}(x^t)\|$,

\begin{equation}\label{diff_G_F}
\begin{aligned}
    & \left|\sum_{i=1}^n\|\nabla_{i}\tilde G_F(x^t)-\nabla_i F_{-i}(x^t)\|-\sum_{i=1}^n\|\nabla_{i} G_F(x^t)-\nabla_i F_{-i}(x^t)\|\right|\\
    &\leq \sum_{i=1}^n \left|\|\nabla_{i}\tilde G_F(x^t)-\nabla_i F_{-i}(x^t)\|-\|\nabla_{i} G_F(x^t)-\nabla_i F_{-i}(x^t)\|\right|\leq\sum_{i=1}^n\|\nabla_{i}\tilde G_F(x^t)-\nabla_{i} G_F(x^t)\|\\
    &\leq\sqrt{n}\sqrt{\sum_{i=1}^n\|\nabla_{i}\tilde G_F(x^t)-\nabla_{i} G_F(x^t)\|^2}\leq \delta\sqrt{n} \sqrt{\sum_{i=1}^n\|\nabla_i f_i(x^t)\|^2}\\
    &\leq\frac{\delta\sqrt{n}\sqrt{\tilde B}}{\gamma-\delta}\leq \frac{\delta\sqrt{n}}{\gamma-\delta}\sum_{i=1}^n\|\nabla_{i}\tilde G_F(x^t)-\nabla_i F_{-i}(x^t)\|\leq \frac{1}{2}\sum_{i=1}^n\|\nabla_{i}\tilde G_F(x^t)-\nabla_i F_{-i}(x^t)\|.
\end{aligned}
\end{equation}

where we apply \cref{b_t_f} in the sixth line and $\delta=\alpha^2\leq\frac{\gamma}{6n}$ in the last line. A direct proposition from \cref{diff_G_F} is

\begin{equation}\label{diff_G_F_pro}
    \sum_{i=1}^n\|\nabla_{i}\tilde G_F(x^t)-\nabla_i F_{-i}(x^t)\|\leq 2\sum_{i=1}^n\|\nabla_{i}\tilde G_F(x^t)-\nabla_i F_{-i}(x^t)\|
\end{equation}

To bound the difference of $k^t=-2+\frac{A}{B}$ and $\tilde k^t=-2+\frac{\tilde A}{\tilde B}$, we need to consider the followings,

\begin{equation}\label{tilde_B_B}
\begin{aligned}
        |\tilde B-B|&=\left|\sum_{i=1}^n\|\nabla_{i}\tilde G_F(x^t)-\nabla_i F_{-i}(x^t)\|^2-\sum_{i=1}^n\|\nabla_{i} G_F(x^t)-\nabla_i F_{-i}(x^t)\|^2\right|\\
        &\leq \sum_{i=1}^n\left|\|\nabla_{i}\tilde G_F(x^t)-\nabla_i F_{-i}(x^t)\|^2-\|\nabla_{i} G_F(x^t)-\nabla_i F_{-i}(x^t)\|^2\right|\\
        &= \sum_{i=1}^n \big|\|\nabla_{i}\tilde G_F(x^t)-\nabla_i F_{-i}(x^t)\|-\|\nabla_{i} G_F(x^t)-\nabla_i F_{-i}(x^t)\|\big|\\
        &\times\big|\|\nabla_{i}\tilde G_F(x^t)-\nabla_i F_{-i}(x^t)\|+\|\nabla_{i} G_F(x^t)-\nabla_i F_{-i}(x^t)\|\big|\\
        &\leq \frac{1}{2}\sum_{i=1}^n\|\nabla_{i}\tilde G_F(x^t)-\nabla_i F_{-i}(x^t)\|(\|\nabla_{i}\tilde G_F(x^t)-\nabla_i F_{-i}(x^t)\|+\|\nabla_{i} G_F(x^t)-\nabla_i F_{-i}(x^t)\|)\\
        &\leq\frac{3\delta}{2(\gamma-\delta)}(\sum_{i=1}^n\|\nabla_{i}\tilde G_F(x^t)-\nabla_i F_{-i}(x^t)\|)^2\leq\frac{3n\delta}{2(\gamma-\delta)}\tilde B\leq \frac{1}{2} \tilde B.
\end{aligned}
\end{equation}

where we apply \cref{diff_G_F} in the fifth line, \cref{diff_G_F_pro} in the sixth line and $\delta=\alpha^2\leq \frac{\gamma}{6n}$ in the last line. Notice that $|k^t-\tilde k^t|=\left|\frac{A}{B}-\frac{\tilde A}{\tilde B}\right|\leq \left|\frac{A}{B}-\frac{ A}{\tilde B}\right|+\left|\frac{A}{\tilde B}-\frac{\tilde A}{\tilde B}\right|$, we show that bound for each part,


\begin{equation}\label{tilde_k_k_1}
\begin{aligned}
    \left|\frac{A}{B}-\frac{ A}{\tilde B}\right|
    &\leq \sqrt{\sum_{i=1}^n\|\nabla_{i} f_{i}(x^t)\|^2}\sqrt{ B}\left|\frac{1}{B}-\frac{1}{\tilde B}\right|=  \sqrt{\sum_{i=1}^n\|\nabla_{i} f_{i}(x^t)\|^2}\sqrt{ B}\frac{1}{B\tilde B}\left|\tilde B-B\right|\\
    &\leq  \sqrt{\sum_{i=1}^n\|\nabla_{i} f_{i}(x^t)\|^2}\sqrt{ B}\frac{3n\delta}{2(\gamma-\delta)B}\leq \frac{3n\delta}{2(\gamma-\delta)^2}\sqrt{\frac{\tilde{B}^t}{B}}\leq  \frac{5n\delta}{(\gamma-\delta)^2}
\end{aligned}
\end{equation}
where we apply Cauchy-Schwarz inequality in the first line, and \cref{tilde_B_B} in the third and last line. Also,

\begin{equation}\label{tilde_k_k_2}
\begin{aligned}
    \left|\frac{A}{\tilde B}-\frac{\tilde A}{\tilde B}\right|&\leq \frac{1}{\tilde B}|A-\tilde A|\leq \frac{1}{\tilde B}|\sum_{i=1}^n\langle \nabla_i f_i(x^t), \tilde\nabla_i G_F(x^t)-\nabla_i G_F(x^t)\rangle|\\
    &\leq  \frac{1}{\tilde B} \sqrt{\sum_{i=1}^n\|\nabla_if_i(x^t)\|^2}\sqrt{\sum_{i=1}^n\|\tilde\nabla_i G_F(x^t)-\nabla_i G_F(x^t)\|^2}\\
    &\leq \frac{1}{\tilde B} \sqrt{\sum_{i=1}^n\|\nabla_if_i(x^t)\|^2}\sqrt{\sum_{i=1}^n\|\tilde\nabla_i G_F(x^t)-\nabla_i G_F(x^t)\|^2}\\
    &\leq\frac{\delta}{\tilde B}\sum_{i=1}^n\|\nabla_if_i(x^t)\|^2\leq \frac{4\delta}{(\gamma-\delta)^2}.
\end{aligned}
\end{equation}

Combining  \cref{tilde_k_k_1} and \cref{tilde_k_k_2}, we get

\begin{equation}\label{tilde_k_k}
\begin{aligned}
    |k^t-\tilde k^t|&\leq \left|\frac{A}{B}-\frac{ A}{\tilde B}\right|+\left|\frac{A}{\tilde B}-\frac{\tilde A}{\tilde B}\right|\leq \frac{(5n+4)\delta}{(\gamma-\delta)^2}\leq \alpha\leq 1.
\end{aligned}
\end{equation}

where we apply $\delta=\alpha^2\leq\frac{\gamma^2C}{4(5n+4)}$ and $C\leq 1$.
To bound $|\tilde k^t|$ and $|k^t|$ by \cref{b_t_f},

\begin{equation}\label{tildekt_max}
|\tilde k^t|=|-2+\frac{\tilde A}{\tilde B}|\leq 2+|\frac{\tilde A}{\tilde B}|\leq 2+\frac{1}{\gamma-\delta}\leq 4,
\end{equation}
and
\begin{equation}\label{kt_max1}
    |k^t|=|k^t-\tilde k^t+\tilde k^t|\leq|k^t-\tilde k^t|+|\tilde k^t|\leq 1+|\tilde k^t|\leq 5,
\end{equation}
where we apply $\delta=\alpha^2\leq\frac{\gamma}{6n}$. 

As a result, it is able to bound the iterates between the ideal case and the approximate case. Consider the following,

\begin{equation}\label{k_t_deri_G}
\begin{aligned}
    &\sum_{i=1}^n\|k^t\nabla_{i} (G_F(x^t)-F_{-i}(x^t))-\tilde k^t(\nabla_{i}\tilde G_F(x^t)-\nabla_i F_{-i}(x^t))\|^2\\
    &\leq 2\sum_{i=1}^n \|k^t\nabla_{i} (G_F(x^t)-F_{-i}(x^t))-\tilde k^t(\nabla_{i} G_F(x^t)-\nabla_i F_{-i}(x^t))\|^2+2\sum_{i=1}^n\|\tilde k^t\nabla_{i} (G_F(x^t)-F_{-i}(x^t))-\tilde k^t(\nabla_{i}\tilde G_F(x^t)-\nabla_i F_{-i}(x^t))\|^2\\
    &\leq 2|k^t-\tilde k^t|^2 \sum_{i=1}^n\|\nabla_{i} G_F(x^t)-\nabla_i F_{-i}(x^t)\|^2+\sum_{i=1}^n2|\tilde k^t|^2\|\nabla G_F(x^t)-\tilde \nabla G_F(x^t)\|^2\\
    &\leq  2\alpha^2 \sum_{i=1}^n\|\nabla_{i} G_F(x^t)-\nabla_i F_{-i}(x^t)\|^2+\sum_{i=1}^n2|\tilde k^t|^2\|\nabla G_F(x^t)-\tilde \nabla G_F(x^t)\|^2\\
    &\leq  \frac{6nL^2\alpha^4}{\mu^2}\sum_{i=1}^n \|\nabla_i f_i(x^t)\|^2+\sum_{i=1}^n2|\tilde k^t|^2\|\nabla G_F(x^t)-\tilde \nabla G_F(x^t)\|^2\\
    &\leq  \frac{6nL^2\alpha^4}{\mu^2}\sum_{i=1}^n \|\nabla_i f_i(x^t)\|^2+ 32\|\nabla G_F(x^t)-\tilde \nabla G_F(x^t)\|^2\\
    &\leq \left(\frac{6nL^2\alpha^4}{\mu^2}+32\delta^2\right)\sum_{i=1}^n \|\nabla_i f_i(x^t)\|^2\leq  \frac{38nL^2\alpha^4}{\mu^2}\sum_{i=1}^n \|\nabla_i f_i(x^t)\|^2.
\end{aligned}
\end{equation}
where in the last inequality we use \cref{G_deri_bound}, \cref{tilde_k_k} and \cref{tildekt_max}.


In the case of one of ideal settings, we need $\alpha$ to satisfy \cref{alphat_f_G1}. However, we only have the estimation $\tilde\nabla G_F(x^t)$. Next, we show that \cref{alphat_f_G1} is satisfied if $\alpha$ is small enough. Then we can make sure the linear convergence of the ideal case and further bound the difference of $f-G_F$ between the ideal case and the practical case.
To make sure the condition of case 2 in the ideal case still satisfies, we get,

\begin{equation*}
\begin{aligned}
    \left(\frac{B}{A}-1\right)^2&=\left(\frac{\tilde B}{\tilde A}-1 +\frac{B}{A}-\frac{\tilde B}{\tilde A}\right)^2\geq \left(\frac{\tilde B}{\tilde A}-1\right)^2-2\left|\frac{\tilde B}{\tilde A}-1\right|\left|\frac{B}{A}-\frac{\tilde B}{\tilde A}\right|\\
    &\geq 2C-2\left|\frac{\tilde B}{\tilde A}-1\right|\left|\frac{B}{A}-\frac{\tilde B}{\tilde A}\right|\geq 2C-\frac{(5n+4)\delta}{(\gamma-\delta)^2}\left|\frac{\tilde B}{\tilde A}-1\right|\geq 2C-\frac{(5n+4)\delta}{(\gamma-\delta)^2}(\left|\frac{\tilde B}{\tilde A}\right|+1)\\
    &\geq 2C-\frac{(5n+4)\delta}{(\gamma-\delta)^2}(\gamma-\delta+1)\geq 2C-\frac{2(5n+4)\delta}{(\gamma-\delta)^2}\geq C.
\end{aligned}
\end{equation*}
where we apply the condition of case 2 in the approximate case and $\delta=\alpha^2\leq\frac{\gamma^2C}{4(5n+4)}$. It indicates that if the condition of the approximate case holds, the condition of the ideal case also holds.

Therefore, we can bound the $F-G_F$ between these two cases at time $t+1$,

\begin{equation}\label{bound_all}
\begin{aligned}
    &\mathbb{E}[F(x^{t+1})-G_F(x^{t+1})|x^t]-\mathbb{E}[F(\bar x^{t+1})-G_F(\bar x^{t+1})|x^t]\\
    &\leq \sum_{i^t=1}^n \left[\langle \nabla_i F(\bar x^{t+1})-\nabla_i G_F(\bar x^{t+1}), x_{i^t}^{t+1}-\bar x_{i^t}^{t+1} \rangle+\frac{n(L+L')}{2}\|x_{i^t}^{t+1}-\bar x_{i^t}^{t+1}\|^2\right]\\
    &\leq \sum_{i^t=1}^n \left[\langle \nabla_i F(x^{t})-\nabla_i G_F(x^{t}), x_{i^t}^{t+1}-\bar x_{i^t}^{t+1} \rangle+\langle \nabla_i F(\bar x^{t+1})-\nabla_i F(x^t)-\nabla_i G_F(\bar x^{t+1})-\nabla_i G_F(x^t), x_{i^t}^{t+1}-\bar x_{i^t}^{t+1} \rangle+\frac{n(L+L')}{2}\|x_{i^t}^{t+1}-\bar x_{i^t}^{t+1}\|^2\right]\\
    &\leq \frac{1}{n}\sum_{i=1}^n \langle \nabla_i F(x^{t})-\nabla_i G_F(x^{t}), \alpha(\tilde k^t(\nabla_{i}\tilde G_F(x^t)-\nabla_i F_{-i}(x^t))-k^t(\nabla_{i} G_F(x^t)-\nabla_i F_{-i}(x^t))) \rangle\\
    &+\frac{1}{n}\sum_{i=1}^n \langle \nabla_i F(\bar x^{t+1})-\nabla_i F(x^t)-\nabla_i G_F(\bar x^{t+1})-\nabla_i G_F(x^t), \alpha(\tilde k^t(\nabla_{i}\tilde G_F(x^t)-\nabla_i F_{-i}(x^t))-k^t(\nabla_{i} G_F(x^t)-\nabla_i F_{-i}(x^t))) \rangle\\
    &+\frac{L+L'}{2}\sum_{i=1}^n\|\alpha(\tilde k^t(\nabla_{i}\tilde G_F(x^t)-\nabla_i F_{-i}(x^t))-k^t(\nabla_{i} G_F(x^t)-\nabla_i F_{-i}(x^t)))\|^2
\end{aligned}
\end{equation}

To bound this, we estimate three parts from the last three lines of \cref{bound_all}. By applying \cref{G_deri_bound} and \cref{k_t_deri_G}, the first term is bounded by 

\begin{equation}
\begin{aligned}
    &\frac{1}{n}\sum_{i=1}^n \langle \nabla_i F(x^{t})-\nabla_i G_F(x^{t}), \alpha(\tilde k^t(\nabla_{i}\tilde G_F(x^t)-\nabla_i F_{-i}(x^t))-k^t(\nabla_{i} G_F(x^t)-\nabla_i F_{-i}(x^t))) \rangle\\
    &\leq \frac{\alpha}{n}\sqrt{\sum_{i=1}^n\|\nabla_i F(x^{t})-\nabla_i G_F(x^{t})\|^2}\sqrt{\sum_{i=1}^n\|\tilde k^t(\nabla_{i}\tilde G_F(x^t)-\nabla_i F_{-i}(x^t))-k^t(\nabla_{i} G_F(x^t)-\nabla_i F_{-i}(x^t))\|^2}\\
    &\leq \frac{\alpha L}{\sqrt{n}\mu}\sqrt{\sum_{i=1}^n\|\nabla_i f_i(x^t)\|^2}\sqrt{\sum_{i=1}^n\|\tilde k^t(\nabla_{i}\tilde G_F(x^t)-\nabla_i F_{-i}(x^t))-k^t(\nabla_{i} G_F(x^t)-\nabla_i F_{-i}(x^t))\|^2}\\
    &\leq \frac{7\alpha^3 L^2}{\mu^2}\sum_{i=1}^n\|\nabla_i f_i(x^t)\|^2
\end{aligned}
\end{equation}

The second term is 

\begin{equation}
\begin{aligned}
    &\frac{1}{n}\sum_{i=1}^n \langle \nabla_i F(\bar x^{t+1})-\nabla_i F(x^t)-\nabla_i G_F(\bar x^{t+1})+\nabla_i G_F(x^t), \alpha(\tilde k^t(\nabla_{i}\tilde G_F(x^t)-\nabla_i F_{-i}(x^t))-k^t(\nabla_{i} G_F(x^t)-\nabla_i F_{-i}(x^t))) \rangle\\
    &\leq\frac{\alpha}{n} \sum_{i=1}^n \|\nabla_i F(\bar x^{t+1})-\nabla_i F(x^t)-\nabla_i G_F(\bar x^{t+1})+\nabla_i G_F(x^t)\|\|\tilde k^t(\nabla_{i}\tilde G_F(x^t)-\nabla_i F_{-i}(x^t))-k^t(\nabla_{i} G_F(x^t)-\nabla_i F_{-i}(x^t))\|\\
    &\leq\alpha(L+L') \sum_{i=1}^n \|\bar x^{t+1}-x^t\|\|\tilde k^t(\nabla_{i}\tilde G_F(x^t)-\nabla_i F_{-i}(x^t))-k^t(\nabla_{i} G_F(x^t)-\nabla_i F_{-i}(x^t))\|\\
    &\leq\alpha^2(L+L') \sum_{i=1}^n \|\nabla_i f_i(x^t)+k^t(\nabla_{i} G_F(x^t)-\nabla_i F_{-i}(x^t))\|\|\tilde k^t(\nabla_{i}\tilde G_F(x^t)-\nabla_i F_{-i}(x^t))-k^t(\nabla_{i} G_F(x^t)-\nabla_i F_{-i}(x^t))\|
\end{aligned}
\end{equation}

\begin{equation}
\begin{aligned}
    &\leq\alpha^2(L+L') \sqrt{\sum_{i=1}^n \|\nabla_i f_i(x^t)+k^t(\nabla_{i} G_F(x^t)-\nabla_i F_{-i}(x^t))\|^2}\sqrt{\sum_{i=1}^n\|\tilde k^t(\nabla_{i}\tilde G_F(x^t)-\nabla_i F_{-i}(x^t))-k^t(\nabla_{i} G_F(x^t)-\nabla_i F_{-i}(x^t))\|^2}\\
    &\leq\alpha^2(L+L') \sqrt{\sum_{i=1}^n 2\|\nabla_i f_i(x^t)\|^2+2\|k^t(\nabla_{i} G_F(x^t)-\nabla_i F_{-i}(x^t))\|^2}\sqrt{\sum_{i=1}^n\|\tilde k^t(\nabla_{i}\tilde G_F(x^t)-\nabla_i F_{-i}(x^t))-k^t(\nabla_{i} G_F(x^t)-\nabla_i F_{-i}(x^t))\|^2}\\
    &\leq\alpha^2(L+L') \sqrt{\sum_{i=1}^n (\frac{150nL^2}{\mu^2}+152)\|\nabla_i f_i(x^t)\|^2}\sqrt{\sum_{i=1}^n\|\tilde k^t(\nabla_{i}\tilde G_F(x^t)-\nabla_i F_{-i}(x^t))-k^t(\nabla_{i} G_F(x^t)-\nabla_i F_{-i}(x^t))\|^2}\\
    &\leq\frac{7\sqrt{n}\alpha^4(L+L')L}{\mu} \sqrt{\frac{150nL^2}{\mu^2}+152}\sum_{i=1}^n \|\nabla_i f_i(x^t)\|^2\leq\frac{105n\alpha^4(L+L')L^2}{\mu^2}\sum_{i=1}^n \|\nabla_i f_i(x^t)\|^2.
\end{aligned}
\end{equation}

The third term is 

\begin{equation}
\begin{aligned}
    &\frac{L+L'}{2}\sum_{i=1}^n\|\alpha(\tilde k^t(\nabla_{i}\tilde G_F(x^t)-\nabla_i F_{-i}(x^t))-k^t(\nabla_{i} G_F(x^t)-\nabla_i F_{-i}(x^t)))\|^2\leq \frac{19n(L+L')L^2\alpha^6}{\mu^2}\sum_{i=1}^n \|\nabla_i f_i(x^t)\|^2.
\end{aligned}
\end{equation}

In conclusion, the linear convergence is still guaranteed, i.e.
\begin{equation}\label{f-G_prac_diff1}
\begin{aligned}
&\mathbb{E}[F(x^{t+1})-G_F(x^{t+1})|x^t]-\mathbb{E}[F(\bar x^{t+1})-G_F(\bar x^{t+1})|x^t]\\
&\leq \left(\frac{7\alpha^3 L^2}{\mu^2}+\frac{105n\alpha^4(L+L')L^2}{\mu^2}+\frac{19n(L+L')L^2\alpha^6}{\mu^2}\right)\sum_{i=1}^n \|\nabla_i f_i(x^t)\|^2\\
&\leq \frac{(L+L')\mu\alpha^2}{8L}\sum_{i=1}^n \|\nabla_i f_i(x^t)\|^2\leq \frac{(L+L')\mu\alpha^2}{4}(F(x^t)-G_F(x^t)).
\end{aligned}
\end{equation}
and
\begin{equation*}
\begin{aligned}
&\mathbb{E}[f(x^{t+1})-G_F(x^{t+1})|x^t]=\mathbb{E}[F(\bar x^{t+1})-G_F(\bar x^{t+1})|x^t]+\mathbb{E}[F(x^{t+1})-G_F(x^{t+1})|x^t]-\mathbb{E}[F(\bar x^{t+1})-G_F(\bar x^{t+1})|x^t]\\
&\leq \Big(1-\frac{(L+L')\mu\alpha^2}{2}\Big)(F(x^t)-G_F(x^t))+\frac{(L+L')\mu\alpha^2}{4}(F(x^t)-G_F(x^t))\\
&\leq \Big(1-\frac{(L+L')\mu\alpha^2}{4}\Big)(F(x^t)-G_F(x^t)).
\end{aligned}
\end{equation*}

where we apply $\alpha\leq\min\{\frac{7(L+L')\mu^2}{24L^2},\sqrt{\frac{\mu^3}{2450nL^3}},(\frac{\mu^3}{480nL^3})^{1/4}\}$.

\textbf{Case 3:} From the smoothness of $F-G_F$ and \cref{bar_x1} with $k^t=-1$, we get
\begin{equation}\label{f-G_kt_11}
\begin{aligned}
    \mathbb{E}[F(\bar x^{t+1})-G_F(\bar x^{t+1})|x^t]&\leq F(x^t)-G_F(x^t)-\frac{1}{n}\Big(\alpha-\frac{nL\alpha^2}{2}-\frac{nL'\alpha^2}{2}\Big)\|\nabla F(x^t)-\nabla G_F(x^t)\|^2\\
    &\leq F(x^t)-G_F(x^t)-\frac{\alpha}{2n}\|\nabla F(x^t)-\nabla G_F(x^t)\|^2.
\end{aligned}
\end{equation}
The second line comes from $\alpha\leq\frac{1}{n(L+L')}$. 
From \cref{G_smooth}, we have
\begin{equation*}
\begin{aligned}
    &\mathbb{E}[F(x^{t+1})-G_F(x^{t+1})|x^t]-\mathbb{E}[F(\bar x^{t+1})-G_F(\bar x^{t+1})|x^t]\\
    &\leq \mathbb{E}\left[\langle\nabla_{i^t} F(\bar x^{t+1})- \nabla_{i^t} G_F(\bar x^{t+1}), {x}_{i^t}^{t+1}- \bar{x}_{i^t}^{t+1}\rangle+\frac{n(L+L')}{2}\|{x}_{i^t}^{t+1}- \bar{x}_{i^t}^{t+1}\|^2|x^t\right]\\
    &= \mathbb{E}\left[\langle\nabla_{i^t} F(\bar x^{t+1})- \nabla_{i^t} G_F(\bar x^{t+1}), \alpha( (\nabla_{i^t}\tilde G_F(x^t)-\nabla_{i^t} F_{-i^t}(x^t))-(\nabla_{i^t} G_F(x^t)-\nabla_i F_{-i^t}(x^t)))\rangle|x^t\right]\\
    &+\mathbb{E}\left[\frac{n(L+L')}{2}\|\alpha( (\nabla_{i^t}\tilde G_F(x^t)-\nabla_{i^t} F_{-i^t}(x^t))-k^t(\nabla_{i^t} G_F(x^t)-\nabla_{i^t} F_{-i^t}(x^t)))\|^2|x^t\right]\\
    &= \mathbb{E}[\langle\nabla_{i^t} F(x^{t})- \nabla_{i^t} G_F(x^{t}), \alpha(\nabla_{i^t}\tilde G_F(x^t)-\nabla_{i^t} G_F(x^t))\rangle|x^t]\\
    &+ \mathbb{E}[\langle\nabla_{i^t} F(\bar x^{t+1})-\nabla_{i^t} F(x^{t})- \nabla_{i^t} G_F(\bar x^{t+1})+\nabla_{i^t} G_F(x^{t}), \alpha(\nabla_{i^t}\tilde G_F(x^t)-\nabla_{i^t} G_F(x^t))\rangle|x^t]\\
    &+\mathbb{E}\left[\frac{L+L'}{2}\|\alpha(\nabla_{i^t}\tilde G_F(x^t)-\nabla_{i^t} G_F(x^t))\|^2|x^t\right],\\
    &= \frac{1}{n}\sum_{i=1}^n\langle\nabla_{i} F(x^{t})- \nabla_{i} G_F(x^{t}), \alpha(\nabla_{i}\tilde G_F(x^t)-\nabla_{i} G_F(x^t))\rangle\\
    &+\frac{1}{n} \sum_{i=1}^n\langle\nabla_{i} F(\bar x^{t+1})-\nabla_{i} F(x^{t})- \nabla_{i} G_F(\bar x^{t+1})+\nabla_{i} G_F(x^{t}), \alpha(\nabla_{i}\tilde G_F(x^t)-\nabla_{i} G_F(x^t))\rangle\\
    &+\frac{L+L'}{2}\left\|\sum_{i=1}^n\alpha(\nabla_{i}\tilde G_F(x^t)-\nabla_{i} G_F(x^t))\right\|^2,\\
\end{aligned}  
\end{equation*}

We find the bound for each terms separately. The first term is
\begin{equation*}
\begin{aligned}
   &\frac{1}{n}\sum_{i=1}^n\langle\nabla_{i} F(x^{t})- \nabla_{i} G_F(x^{t}), \alpha(\nabla_{i}\tilde G_F(x^t)-\nabla_{i} G_F(x^t))\rangle,\\
   &\leq\frac{\alpha}{n} \|\nabla F(x^{t})- \nabla G_F(x^{t})\| \|\nabla\tilde G_F(x^t)-\nabla G_F(x^t)\|,\\
   &\leq \frac{\alpha\delta}{n}\|\nabla F(x^{t})- \nabla G_F(x^{t})\|\sqrt{\sum_{i=1}^n\|\nabla_i f_i(x^t)\|^2}
\end{aligned}
\end{equation*}
The second term is 
\begin{equation*}
\begin{aligned}
    &\frac{1}{n} \sum_{i=1}^n\langle\nabla_{i} F(\bar x^{t+1})-\nabla_{i} F(x^{t})- \nabla_{i} G_F(\bar x^{t+1})+\nabla_{i} G_F(x^{t}), \alpha(\nabla_{i}\tilde G_F(x^t)-\nabla_{i} G_F(x^t))\rangle\\
    &\leq\frac{\alpha}{n}\|\nabla F(\bar x^{t+1})-\nabla F(x^{t})- \nabla G_F(\bar x^{t+1})+\nabla G_F(x^{t})\|\|\nabla\tilde G_F(x^t)-\nabla G_F(x^t)\|,\\
    &\leq \alpha(L+L')\|\bar{x}^{t+1}-x^t\|\|\nabla\tilde G_F(x^t)-\nabla G_F(x^t)\|,\\
    &\leq \alpha^2(L+L')\|\nabla F(x^t)-\nabla G_F(x^t)\|\|\nabla\tilde G_F(x^t)-\nabla G_F(x^t)\|,\\
    &\leq \alpha^2 (L+L')\delta \|\nabla F(x^t)-\nabla G_F(x^t)\|\sqrt{\sum_{i=1}^n\|\nabla_i f_i(x^t)\|^2}.
\end{aligned}
\end{equation*}
The third term is 
\begin{equation*}
\begin{aligned}
    &\frac{L+L'}{2}\|\sum_{i=1}^n\alpha(\nabla_{i}\tilde G_F(x^t)-\nabla_{i} G_F(x^t))\|^2\leq \frac{(L+L')\alpha^2\delta^2}{2}\sum_{i=1}^n\|\nabla_i f_i(x^t)\|^2 .
\end{aligned}
\end{equation*}
 Overall, we obtain
\begin{equation*}
\begin{aligned}
    &\mathbb{E}[F(x^{t+1})-G_F(x^{t+1})|x^t]-\mathbb{E}[F(\bar x^{t+1})-G_F(\bar x^{t+1})|x^t]\\
    &\leq (\frac{\alpha\delta}{n}+\alpha^2 (L+L')\delta)\|\nabla F(x^{t})- \nabla G_F(x^{t})\|\sqrt{\sum_{i=1}^n\|\nabla_i f_i(x^t)\|^2}+\frac{(L+L')\alpha^2\delta^2}{2}\sum_{i=1}^n\|\nabla_i f_i(x^t)\|^2 .
\end{aligned}
\end{equation*}
and,
\begin{equation*}
\begin{aligned}
&\mathbb{E}[F(x^{t+1})-G_F(x^{t+1})|x^t]\\
&=\mathbb{E}[F(\bar x^{t+1})-G_F(\bar x^{t+1})|x^t]+\mathbb{E}[F(x^{t+1})-G_F(x^{t+1})|x^t]-\mathbb{E}[F(\bar x^{t+1})-G_F(\bar x^{t+1})|x^t]\\
&\leq F(x^{t})-G_F(x^{t})-\frac{1}{2n}\alpha\|\nabla F(x^t)-\nabla G_F(x^t)\|^2+\left(\frac{\alpha\delta}{n}+\alpha^2 (L+L')\delta\right)\|\nabla F(x^{t})- \nabla G_F(x^{t})\|\sqrt{\sum_{i=1}^n\|\nabla_i f_i(x^t)\|^2}+\frac{(L+L')\alpha^2\delta^2}{2}\sum_{i=1}^n\|\nabla_i f_i(x^t)\|^2\\
&\leq F(x^{t})-G_F(x^{t})-\frac{\alpha}{4n}\|\nabla F(x^t)-\nabla G_F(x^t)\|^2\leq F(x^{t})-G_F(x^{t})-\frac{\alpha\nu}{2n}(F(x^{t})-G_F(x^{t}))^{\frac{2}{\theta}}.
\end{aligned}
\end{equation*}
where we apply $\delta=\alpha^2\leq\min\{\frac{\|\nabla F(x^t)-\nabla G_F(x^t)\|}{12\sqrt{\sum_{i=1}^n\|\nabla_i f_i(x^t)\|^2}},\frac{1}{n(L+L')}\}$. 
From Lemma 6 of \cite{fatkhullin2022sharp}, we have

\begin{equation*}
     \mathbb{E}[f(x^{t+k})-G_F(x^{t+k})|x^t]\leq \frac{(4n)^{\frac{\theta}{2-\theta}}{\frac{2-\theta}{\theta}}^{-\frac{\theta+2}{2-\theta}}+(2n)^\frac{\theta}{2-\theta}\theta^{-\frac{\theta}{2-\theta}}+({\nu\alpha})^{\frac{\theta}{2-\theta}}(f(x^{t})-G_F(x^{t}))}{({\nu\alpha}(k+1))^{\frac{\theta}{2-\theta}}}.
\end{equation*}
\end{proof}

\section{Details of the Application Section}\label{sec:add_app}

\subsection{Proof of $N$-sided PL condition for multi-player Linear Quadratic Game (LQR)}\label{proof_n_sided_LQR}
The system can be written down as
\begin{equation*}
s_{t+1}=As_t+\sum_{i=1}^NB_iu_t^i=As_t+\sum_{i=1}^NB_iK_t^is_t=(A-\sum_{j\neq l}B_jK_j)s_t+B_lK_ls_t,
\end{equation*}
and the system can be written down as
\begin{equation*}
\begin{aligned}
        f_i(K_i,K_{-i})&=\mathbb{E}_{x_0\sim\mathcal{D}}\left[\sum_{t=0}^{+\infty}[(s_t)^T Q_i s_t+(K_is_t)^TR_i K_is_t]\right]
\end{aligned}
\end{equation*}
Define $\Sigma_{K}$ as the state correlation matrix, i.e.
\begin{equation*}
    \Sigma_K=\mathbb{E}_{x_0\sim \mathcal{D}}\sum_{t=0}^\infty s_t(s_t)^T.
\end{equation*}
From the Corollary 5 of \cite{fazel2018global}, we have
\begin{equation*}
    f_i(K_i,K_{-i})-\min_{K_i'}f_i(K_i',K_{-i})\leq\frac{\left\|\Sigma_{K_{i,K_{-i}}^\star,K_{-i}}\right\|}{\sigma_{min}(\Sigma_0)^2\sigma_{min}(R_l)}\|\nabla_{K_i}f_i(K_i,K_{-i})\|_F^2, \forall l
\end{equation*}
where $K_{i,K_{-i}}^\star\in\mathrm{argmin}_{K_i'}f_i(K_i',K_{-i})$.
 Since $K$ is bounded and $\sigma_{min}(\Sigma_0)>0$, then $0<\kappa<+\infty$, and $f$ satisfies $N$-sided PL condition.


\subsection{Counterexample of Multi-convexity for $N$-player LQ Game}\label{ce_LQR_convex}

Here, we only need to prove that there exists $K_1$, $K_1'$ and $K_2$ such that
\begin{equation*}
    f_1(K_1,K_2)+f_1(K_1',K_2)\leq 2f_1(\frac{K_1+K_1'}{2},K_2).
\end{equation*}
where $f_1(K_1, K_2)$ is the objective function of the 2-player linear quadratic game.
 We denote $A$ and $B$ to be $3\times 3$ identity matrix and
\begin{equation*}
K_1=
\begin{bmatrix}
0 & 0 & -10\\
-1 & 0 & 0\\
0 & 0 & 0
\end{bmatrix}
\mathrm{and}\
K_1'=
\begin{bmatrix}
0 & -10 & 0\\
0 & 0 & 0\\
-1 & 0 & 0
\end{bmatrix}
\mathrm{and}\
K_2=
\begin{bmatrix}
1 & 0 & 0\\
0 & 1 & 0\\
0 & 0 & 1
\end{bmatrix}.
\end{equation*}
The matrices $A-B(K_1+K_2)$ and $A-B(K_1'+K_2)$ are both stable, however, the matrix $A-B(\frac{K_1+K_2}{2})$ is unstable. As a result, the objective function $f_1(K_1,K_2),f_1(K_1',K_2)<+\infty$ and $f_1(\frac{K_1+K_1'}{2},K_2)=+\infty$.

\subsection{Additional Experiment}
\paragraph{\textbf{Competitive Resource Allocation Game:}}
Consider a two-player game where two companies, A and B, are competing for resources. In this case, they are both trying to minimize their cost of using the resource, denoted by $x_A$ and $x_B$, but the cost depends on how much of the resource each company uses \citet{roughgarden2009algorithmic}. 
Let us assume the following quadratic cost functions: 
\begin{align*}
f_A(x_A,x_B)=x_A^2-2x_A x_B,\quad
 f_B(x_A,x_B)=x_B^2-2x_A x_B.
\end{align*}

The quadratic terms represent the individual costs of using the resources and the multiplicative term indicates the interaction cost between the two companies. Total costs of this game is $f_A+f_B$.
It is straightforward to see that this objective is 2-sided PL but is not lower bounded. The only NE point is $(0,0)$, which is a strict saddle point. 
Figure \ref{fig:non_convergence} shows the convergence results of A-RBCD and BCD with 100 random initializations. The iterates of A-RBCD always converge to the NE at a linear rate, while BCD diverges. 

\begin{figure*}
\centering
\subfigure[]{
\includegraphics[width=0.5\linewidth,height=4.5cm,trim=2.7cm 9.2cm 2.7cm 9.2cm, clip]{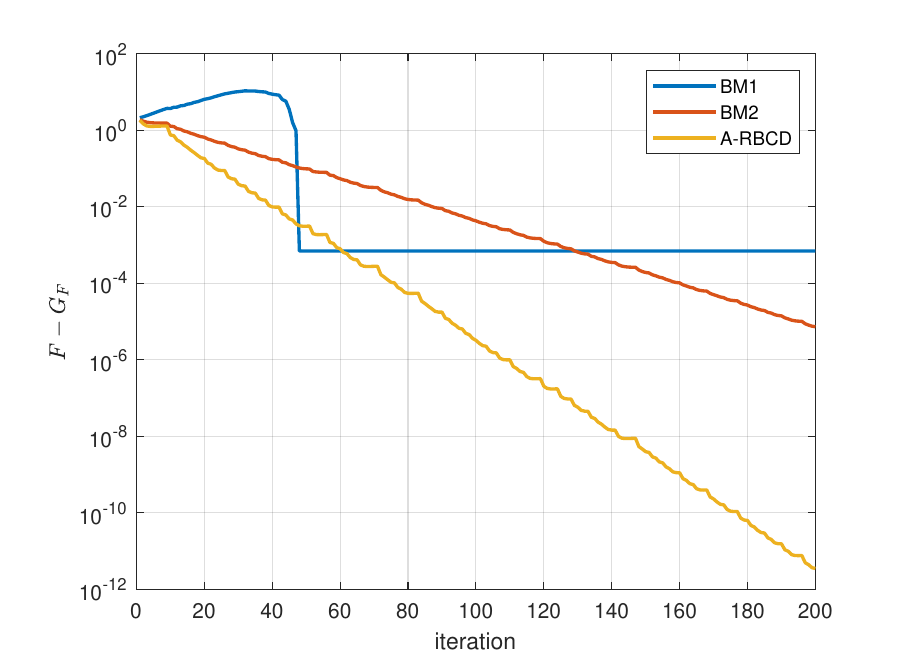}}
\caption{ Performances of A-RBCD, BM1 and BM2 on function $f_A+f_B$.}\label{fig:non_convergence1}
\end{figure*}